\documentclass{article}[11pt]
	\usepackage{a4,geometry}
\usepackage{graphicx}
\usepackage{amsmath,amssymb,amsthm,mathtools}
\usepackage{paralist}
\usepackage{bm}
\usepackage{xspace}
\usepackage{url}
\usepackage{fullpage, prettyref}
\usepackage{boxedminipage}
\usepackage{wrapfig}
\usepackage{ifthen}
\usepackage{color}
\usepackage{xcolor}
\usepackage{framed}
\usepackage{algorithm}
\usepackage[pagebackref,letterpaper=true,colorlinks=true,pdfpagemode=none,urlcolor=blue,linkcolor=blue,citecolor=violet,pdfstartview=FitH]{hyperref}
\usepackage{fullpage}
\usepackage[noend]{algpseudocode}
\usepackage{enumitem}
\usepackage{thmtools}
\usepackage{thm-restate,cleveref}
\newtheorem{theorem}{Theorem}[section]

\newtheorem{lemma}[theorem]{Lemma}
\newtheorem{claim}[theorem]{Claim}

\newtheorem{corollary}[theorem]{Corollary}
\newtheorem{definition}{Definition}

\newtheorem{remark}{Remark}

\newcommand{\ignore}[1]{}

\newcommand{\cC}{{\cal C}}

\newcommand{\cJ}{{\cal J}}

\newcommand{\R}{\mathbb R}

\newcommand{\F}{\mathbb F}
\newcommand{\Z}{{\mathbb Z}}
\newcommand{\eps}{\varepsilon}

\newcommand{\poly}{\mathrm{poly}}

\newcommand{\calL}{{\cal L}}
\newcommand{\calF}{{\cal F}}

\newcommand{\calH}{{\cal H}}

\newcommand{\calP}{{\cal P}}

\newcommand{\calS}{{\cal S}}
\newcommand{\calT}{{\cal T}}

\newcommand{\calC}{{\cal C}}

\newcommand{\calI}{{\cal I}}
\newcommand{\calJ}{{\cal J}}
\newcommand{\calR}{{\cal R}}

\newcommand{\br}{\boldsymbol{r}}

\newcommand{\ceil}[1]{\lceil#1\rceil}

\newcommand{\floor}[1]{\lfloor#1\rfloor}

\newcommand{\Sec}[1]{\hyperref[sec:#1]{\S\ref*{sec:#1}}} %section
\newcommand{\Eqn}[1]{\hyperref[eq:#1]{(\ref*{eq:#1})}} %equation
\newcommand{\Fig}[1]{\hyperref[fig:#1]{Fig.\,\ref*{fig:#1}}} %figure
\newcommand{\Tab}[1]{\hyperref[tab:#1]{Tab.\,\ref*{tab:#1}}} %table
\newcommand{\Thm}[1]{\hyperref[thm:#1]{Theorem\,\ref*{thm:#1}}} %theorem
\newcommand{\Fact}[1]{\hyperref[fact:#1]{Fact\,\ref*{fact:#1}}} %fact
\newcommand{\Lem}[1]{\hyperref[lem:#1]{Lemma\,\ref*{lem:#1}}} %lemma
\newcommand{\Prop}[1]{\hyperref[prop:#1]{Prop.~\ref*{prop:#1}}} %property
\newcommand{\Cor}[1]{\hyperref[cor:#1]{Corollary~\ref*{cor:#1}}} %corollary
\newcommand{\Conj}[1]{\hyperref[conj:#1]{Conjecture~\ref*{conj:#1}}} %conjecture
\newcommand{\Def}[1]{\hyperref[def:#1]{Definition~\ref*{def:#1}}} %definition
\newcommand{\Alg}[1]{\hyperref[alg:#1]{Alg.~\ref*{alg:#1}}} %algorithm
\newcommand{\Ex}[1]{\hyperref[ex:#1]{Ex.~\ref*{ex:#1}}} %example
\newcommand{\Clm}[1]{\hyperref[clm:#1]{Claim~\ref*{clm:#1}}} %example

\def\diam{\mathrm{diam}}
\def\dist{\mathrm{dist}}
\def\mckc{{\sffamily Heterogeneous Cap-$k$-Center}\xspace}
\def\cckp{$Q|f_i|C_{min}$\xspace}
\def\opt{\mathsf{OPT}}

\def\x{{\mathsf x}}
\def\y{y^{{\mathsf{int}}}}
\def\z{\bar{z}}
\def\Supp{\mathsf{Supp}\xspace}

\def\effc{c_{\mathrm{eff}}}
\def\zz{z^{\mathsf{int}}}
\def\suff{\mathsf{suff}}

\def\Cb{C_{\mathsf{blue}}}
\def\Cbb{C_{\mathsf{black}}}

\def\Cd{C_{\mathsf{del}}}

\newcommand{\brp}{{(p)}}
\renewcommand{\br}[1]{{(#1)}}
\newcommand{\bc}{{\bar c}}
\renewcommand{\epsilon}{\varepsilon}

\DeclareMathOperator*{\argmax}{arg\,max}
\DeclareMathOperator*{\argmin}{arg\,min}

\newcommand{\initOneLiners}{%
    \setlength{\itemsep}{0pt}
    \setlength{\parsep }{0pt}
    \setlength{\topsep }{0pt}
}
\newenvironment{oneLiners}[1][\ensuremath{\bullet}]
    {\begin{list}
        {#1}
        {\initOneLiners}}
    {\end{list}}

\usepackage[colorinlistoftodos,textsize=tiny,textwidth=2cm,color=red!25!white]{todonotes}

\begin{document}
\title{The Heterogeneous Capacitated $k$-Center Problem}
\date{}
\author{Deeparnab Chakrabarty \\ \small Microsoft Research India\\\small  deeparnab@gmail.com \and Ravishankar Krishnaswamy \\ \small Microsoft Research India\\\small  ravishankar.k@gmail.com \and Amit Kumar\thanks{Part of the work done while visiting Microsoft Research, India} \\ \small  Comp. Sci. \& Engg., IIT Delhi \\ \small amitk@cse.iitd.ac.in}
\maketitle
\begin{abstract}
	In this paper we initiate the study of the {\em heterogeneous capacitated $k$-center problem}: given a metric space $X = (F \cup C, d)$, and a collection of capacities. The goal is to open each capacity at a unique facility location in $F$, and also to assign clients to facilities so that the number of clients assigned to any facility is at most the capacity installed; the objective is then to minimize the maximum distance between a client and its assigned facility. If all the capacities $c_i$'s are identical, the problem becomes the well-studied {\em uniform capacitated $k$-center problem} for which constant-factor approximations are known~\cite{Bar-IlanKP93,KhullerS00}.
%However, the non-uniform generalization (called the {\em non-uniform capacitated $k$-center problem}) of this problem which has received much attention recently differs from our problem in the following sense: in the non-uniform problem, each facility location $f$ has a fixed capacity $c_f$, and the goal is to open $k$ centers and assign clients so that the capacity constraints are satisfied. On the other hand, in our {\em heterogeneous} problem, a set of $k$ different capacities is specified, and the goal is to suitably place them in facility locations.
The additional choice of determining which capacity should be installed in which location makes our problem considerably different from this problem, as well the non-uniform generalizations studied thus far in literature. In fact, one of our contributions is in relating the heterogeneous problem to special-cases of the classical {\em santa-claus problem}. Using this connection, and by designing new algorithms for these special cases, we get the following results for \mckc.
\begin{oneLiners}
\item A quasi-polynomial time $O(\log n/\epsilon)$-approximation where every capacity is violated by $1+\epsilon$.
\item A polynomial time $O(1)$-approximation where every capacity is violated by an $O(\log n)$ factor.
\end{oneLiners}
We get improved results for the {\em soft-capacities} version where we can place multiple facilities in the same location.
\end{abstract}

\section{Introduction}
The capacitated $k$-center problem is a classic optimization problem where a finite metric space $(X,d)$ needs to be partitioned into $k$ clusters so that  every  cluster has cardinality at most
some specified value $L$, and the objective is to minimize the maximum intra-cluster distance. This problem introduced by Bar-Ilan et al~\cite{Bar-IlanKP93} has many applications~\cite{LuptonMY98,MorganL77,Murthy1983AnAA}. %see khuller-sussmann paper for these refs.
One application is deciding placement of machine locations (centers of clusters) in a network scheduling environment where jobs arise in a metric space and the objective function has a  job-communication (intra-cluster distance) and machine-load (cardinality)
component~\cite{PSW97}. %\mcomment{Add Venkat's References?}
 The above problem is {\em homogeneous} in the sizes of the clusters, that is, it has the same cardinality constraint $L$ for each cluster. In many applications, one would ask for a \emph{heterogeneous} version of the problem where we have a different cardinality constraint for the clusters.
For instance in the network scheduling application above, suppose we had machines of differing speeds. We could possibly load higher-speed machines with more jobs than lower-speed ones. In this paper, we study  this heterogenous version.

\begin{definition}\emph{(The \mckc Problem\footnote{Technically, we should call our problem the Heterogeneous Capacitated $k$-Supplier Problem since we can only open centers in $F$. However, we avoid making this distinction throughout this paper.}.)}
	%	\begin{itemize}[noitemsep]
	%		\item {\bf Input:} Metric space $\left(X=F\cup C, d\right)$ where $F$ are facilities and $C$ are clients.
	%		\\$~~~~~~~~~~$ Collection: $(k_1,c_1), \ldots, (k_P,c_P)$ of $P$-tuples of positive integers.
	%		\item {\bf Output:} $F_1,\ldots, F_P \subseteq F$ which are pairwise disjoint and $|F_p| \leq k_p$ for all $1\le p\leq P$. \\
	%		$~~~~~~~~~~~~$ Assignment $\phi: C\to F_1\cup F_2\cup \cdots \cup F_P$ such that for all $p$, $|\{j\in C: \phi(j)\in F_p\}|\leq c_p$.
	%		\item {\bf Objective:} Minimize $\max_{j\in C} d(j,\phi(j))$.
	%	\end{itemize}
	We are given a metric space $(X = F\cup C,d)$  where $C$ and $F$ represent the clients and facility locations.%where the set $X$ is partitioned into facilities $F$ and clients $C$.
	%Furthermore, the input contains
	We are also given a collection of {\em heterogeneous} capacities: $(k_1,c_1), (k_2,c_2),\ldots, (k_P,c_P)$ with $k_i$ copies of capacity $c_i$.
	%with $c_1 \leq c_2 \le \cdots \le c_t$,  to indicate we can open $k_p$ centers (called type $p$ centers) with capacity $c_p$ in $F$.
	The objective is to install these capacities at unique locations $F'\subseteq F$, and find an assignment $\phi:C\to F'$ of clients to these locations,
	such that for any $i\in F'$ the number of clients $j$ with $\phi(j) = i$ is at most the capacity installed at $i$, and $\max_{j\in C} d(j,\phi(j))$ is minimized.
%	The objective is to minimize
	%
	% to open these centers and assign all points of $C$ to one of these so that (a) any center of type  $p$ serves at most $c_p$ clients, and (b) the maximum distance of a client $j$ to its assigned center $i$ is minimized.
%	We use $\opt$ to denote this latter distance of the optimum solution.
	A weaker version, which we call \mckc  with soft capacities, allows multiple capacities to be installed at the same location.
\end{definition}
\noindent
Note that when all $c_p = L$ and $\sum_p k_p = k$, we get back the usual capacitated $k$-center problem.
The \mckc problem is relevant in many applications where the resources available are heterogenous. The machine placement problem was one example which has applications in network scheduling~\cite{QiuSZ15,ImM15a} and distributed databases~\cite{MorganL77,SKRN15}. Another example is that of  vehicle routing problems with  fleets of different speeds~\cite{GortzMN016}. A third relevant application may be clustering; often clusters of equal sizes are undesirable~\cite{GuhaRS01} and explicitly introducing heterogeneous constraints might lead to desirable clusters.
In this paper, we investigate the worst-case complexity of the \mckc problem. %\smallskip

Bar-Ilan et al~\cite{Bar-IlanKP93} gave a $10$-approximation for the homogeneous capacitated $k$-center problem which was improved to a $6$-factor approximation by Khuller and Sussmann~\cite{KhullerS00}. One cannot get a better than $2$-approximation even for the {\em uncapacitated} $k$-center problem~\cite{HochbaumS85}. More recently, the {\em non-uniform} capacitated $k$-center problem was considered~\cite{CyganHK12,AnBCGMS14} in the literature:~in this problem every facility $v\in F$ has a pre-determined capacity $c_v$ if opened (and $0$ otherwise). We remark that the non-uniform version and our heterogeneous version seem unrelated in the sense that none is a special case of the other.
Cygan et al~\cite{CyganHK12} gave an $O(1)$-approximation for the problem which was improved to a $11$-approximation by An et al~\cite{AnBCGMS14}.
%Cygan and Kociumura~\cite{CKstacs14} look at the capacitated $k$-center with outliers problem where some $z$ clients need not be assigned; note that this  is a special case of the \mckc problem with $P=1$ where we allow $z$ clients with capacity $1$.

%\paragraph{Results}

%
%\subsection{The \mckc Problem}
%
% When there is only one type of center with capacity $c_1 = c$ and $n_1 = k$, we obtain the uniform capacitated $k$-center problem~\cite{barilan,khuller-sussman}. There are $O(1)$-approximation algorithms for this problem. We note that the non-uniform capacitated $k$-center problem which has more recently  been studied~\cite{cygan,ola,auonon} in the literature seems unrelated to the \mckc problem. (also add cygan-kociumaka ref)\smallskip
%
\subsection*{Connection to Non-Uniform Max-Min Allocation Problems.}
One main finding of this paper is the connection of the \mckc problem to {\em non-uniform} max-min allocation (also known as Santa Claus~\cite{BansalS06}) problems, which
 underscores its difficulty and difference from  the homogeneous capacitated $k$-center problems. %we relate it to the classic $3$-Partitioning problem~\cite{Garey-Johnson}: in this problem
We use the machine scheduling parlance to describe the max-min allocation problems.

\begin{definition}[$Q||C_{min}$ and \cckp]
In the\footnote{(Ab)using Graham's notation} $Q||C_{min}$ problem, one is given $m$ machines with demands $D_1,\ldots,D_m$ and $n$ jobs with capacities $c_1,\ldots,c_n$,
and the objective is to find an assignment of the jobs to machines satisfying each demand.
%We call this problem, (ab)using Graham's notation as $Q||C_{min}$.
In the {\em cardinality constrained}  non-uniform max-min allocation problem, denoted as the \cckp problem, each machine further comes with a cardinality constraint $f_i$, and a feasible solution cannot allocate more than $f_i$ jobs to machine $i$. The objective remains the same. An $\alpha$-approximate feasible solution assigns each machine $i$  total capacity at least $ D_i/\alpha$.
\end{definition}
\noindent
We now show how these problems arise as special cases
of the \mckc problem, even with soft capacities.
\begin{remark}[Reduction from $Q|f_i|C_{min}$]\label{frem:cckp}
	Given an instance $\calI$ of $Q|f_i|C_{min}$, construct the instance of \mckc as follows. The capacities available to us are precisely the capacities of the jobs in $\calI$.
	The metric space is divided into $m$ groups $(F_1\cup C_1),\ldots,(F_m\cup C_m)$ such that the distance between nodes in any group is $0$ and across groups is $1$.
	Furthermore, for $1\leq i\leq m$, $|F_i| = f_i$ and $|C_i| = D_i$. Observe that the \mckc instance has a $0$-cost, capacity-preserving solution iff $\calI$ has a feasible assignment.
%	When there are soft capacities, the above reduction is from $Q||C_{min}$.
\end{remark}
The $Q||C_{min}$ and $Q|f_i|C_{min}$ problems are strongly NP-hard.\footnote{A simple reduction from 3-dimensional matching shows NP-hardness of \cckp and $Q || C_{min}$ even when the demands and capacities are polynomially bounded.} Therefore, no non-trivial approximation to \mckc, even the soft-version, exists unless we {\em violate} the capacities.
This observation, which is in contrast to the homogeneous version, motivates us to look at bicriteria approximation algorithms.
%
%
%
%
%we are given $3t$ non-negative numbers $\{a_1,\ldots,a_{3t}\}$ summing to $Dt$, and we have to decide if there is a partition into $t$-groups $S_1,\ldots, S_t$ such that $|S_i| =  3$ and $\sum_{j\in S_i} a_j = D$ for all $i$.
%
%Given an instance of $3$-Partition, consider an instance of \mckc as follows. We let $n_i = 1$ and $c_i = a_i$ for $1\leq i\leq 3t$. %Also, let $D:= \sum_{i=1}^{3t} a_i$ which we assume wlog is divisible by $3$.
%Consider a metric space where $X = F\cup C$ where $X$ is partitioned into $X_1 = (F_1\cup C_1),\ldots,X_t = (F_t\cup C_t)$ such that $|F_i| = 3$ and $|C_i| = D$ for all $i$.
%Furthermore, for any pair of points $u,v$ in the same $X_i$ their distance is $0$ and otherwise $\infty$.
%Now observe that $\opt$ for the \mckc instance is {\em finite} if and only if the $3$-Partitioning instance is a Yes-instance. In other words, unless $P=NP$, there can be no approximation algorithm for the problem
%unless we allow some capacity violation.
\begin{definition}[$(a,b)$-Bicriteria Approximation.]
	Given an instance of the \mckc problem, an $(a,b)$-approximate feasible solution installs $k_p$ units of $c_p$ capacity, and
	assigns clients to  facilities at most $a\cdot\opt$ away and the number of clients assigned to a facility where a capacity
	$c_p$ has been opened\footnote{We add the ceiling to avoid pesky rounding issues.} is $\leq \ceil{bc_p}$. An $(a,b)$-bicriteria approximation algorithm always returns an $(a,b)$-approximate feasible solution.
\end{definition}
Although bicriteria approximation algorithms may be unsatisfactory, sometimes these can give unicriteria approximations for other related problems.
We mention one application that  we alluded to above, and in fact was the starting point of this research,  which may be of independent interest.
\begin{definition}\emph{(Machine Placement Problem for Network Scheduling.)}\label{fdef:mpp}
	The input is a metric space $(X=F\cup C,d)$ with jobs with processing times $p_j$ at locations $C$. We are also given $P$ machines with speeds $s_1,s_2,\ldots,s_P$.
	The goal is to find a placement of these machines on $F$ and schedule the jobs on these machines so as to minimize the makespan. A job can be scheduled on a machine only after it reaches the location of the machine. In the ``soft" version of the problem, multiple machines may be placed in the same location.
\end{definition}
Although we do not prove it in this paper, any $(a,b)$-bicriteria approximation algorithm for (soft) \mckc problem implies an $O(a+b)$ approximation for the (soft) machine placement problem.

\subsection{Results}
The reduction in Remark~\ref{frem:cckp} does not rule out arbitrarily small violations to the capacity. Indeed the $Q||C_{min}$ problem has a PTAS~\cite{AzarE98}.
Our first couple of results give logarithmic approximation to the cost with $(1+\epsilon)$-violations to the capacities.
\begin{theorem}\label{fthm:1}
	Fix an $\eps > 0$. There exists an $(O(\log n/\eps), (1+\eps))$-bicriteria approximation algorithm for the \mckc problem running in time $C_\eps^{\tilde{O}(\log^3n)}$ for a constant $C_\epsilon$ depending only on $\epsilon$. There exists an $(O(\log n/\eps), (1+\eps))$-bicriteria approximation algorithm for the \mckc problem with soft capapcities running in time $n^{O(1/\epsilon)}$.
\end{theorem}
We are not aware of non-trivial results for the \cckp problem (although, see Remark~\ref{frem:cckp-prev} below). We therefore call out the special case of the above theorem.
This makes it rather improbable for \cckp to be APX-hard, and we leave the design of a PTAS as a challenging open problem.
\begin{theorem}\label{fthm:q}
	There is a QPTAS for the \cckp problem.
\end{theorem}

Our main technical meat of the paper is in reducing the logarithmic factor in the approximation to the distance.
We can give $O(1)$-approximations if the violations are allowed to be $O(1)$ in the soft-capacity case and $O(\log n)$ in the general case.
These algorithms run in polynomial time.
\begin{theorem}\label{fthm:2}
	%For any $\eps > 0$, there exists an $\left(\tilde{O}(1/\eps), (1+\eps)\right)$-bicriteria approximation algorithm for \mckc which runs in time $n^{\tilde{O}\left(\frac{\log n}{\eps}\right)}$.
	There is a polynomial time  $(O(1),O(\log n))$-bicriteria approximation algorithm for the \mckc problem.
\end{theorem}
\begin{theorem}\label{fthm:2a}
	%For any $\eps > 0$, there exists an $\left(\tilde{O}(1/\eps), (1+\eps)\right)$-bicriteria approximation algorithm for \mckc which runs in time $n^{\tilde{O}\left(\frac{\log n}{\eps}\right)}$.
	For any $\delta>0$, there is a polynomial time  $(\tilde{O}(1/\delta),2+\delta)$-bicriteria approximation algorithm for the \mckc problem with soft capacities.
\end{theorem}
In particular we have polynomial time $O(1)$ and $O(\log n)$ approximation algorithms for the machine placement problem of Definition~\ref{fdef:mpp}.
Once again, we call out what we believe is the first polynomial time non-trivial approximation to \cckp.
\begin{theorem}\label{fthm:cckp}
	There is a polynomial time logarithmic approximation algorithm for the \cckp problem.
\end{theorem}
We end the section by stating what we believe was the frontier of knowledge for the \cckp problem.
\begin{remark}[Known algorithms for \cckp]\label{frem:cckp-prev}\emph{
	To our knowledge, \cckp has not been explicitly studied in the literature. % though the min-max version has been
	However, in a straightforward manner one can reduce \cckp to {\em non-uniform}, restricted-assignment max-min allocation problem (which we denote as $Q|restr|C_{min}$) where,  instead of the cardinality constraint dictated by $f_i$, we restrict jobs to be assigned only to a subset of the machines:
	for every machine $i$ and job $j$, $j$ can be assigned to $i$ iff $c_j \geq D_i/2f_i$. It is not hard to see that a $\rho$-approximation for the $Q|restr|C_{min}$ implies a $2\rho$-approximation for the \cckp instance.
}

\emph{
Clearly $Q|restr|C_{min}$  is a special case of the general max-min allocation problem~\cite{ChakrabartyCK09} and therefore for any $\epsilon>0$, there are $n^{O(1/\epsilon)}$-time algorithms
achieving $O(n^\epsilon)$-approximation. We do not know of any better approximations for $Q|restr|C_{min}$. The so-called Santa Claus problem is the {\em uniform} version $P|restr|C_{min}$ where all demands are the same~\cite{BansalS06}. This has a $O(1)$-approximation algorithm~\cite{Feige08,AsadpourFS12,PolacekS16}. However all these algorithms use the configuration LP; unfortunately for the non-uniform version $Q|restr|C_{min}$, the configuration LP has an integrality gap of $\Omega(\sqrt{n})$ (this example is in fact the same example of~\cite{BansalS06} proving the gap for general max-min allocation -- see Appendix~\ref{sec:app-bsig}.)}
\end{remark}

\subsection{Outline of Techniques} \label{fsec:overview}
We give a brief and informal discussion of how we obtain our results, referring to the formal definition whenever needed.
In a nutshell, we obtain our results by reducing the \mckc problem {\em to} the \cckp problem (complementing the reduction {\em from} discussed in Remark~\ref{frem:cckp}).
We provide two reductions -- the first incurs logarithmic approximation to the cost but uses black-box algorithms for \cckp, the second incurs $O(1)$-approximation to the cost but uses ``LP-based'' algorithms for \cckp.
Both these reductions proceed via {\em decomposing} the given instance of the \mckc problem. %\smallskip

\medskip \noindent {\bf Warm-up: Weak Decompostion.}
Given a \mckc instance, suppose we \emph{guess} the optimal objective value, which we can assume to be $1$ after scaling. Then, we construct a graph connecting client $j$ with facility location $i$ iff $d(i,j) \leq 1$.
Then, starting at an arbitrary client and using a simple region-growing technique (like those used for the graph cut problems~\cite{LeightonR99,GargVY96}), we can find a set of clients $J_1$ of along with their neighboring facility locations $T_1 = \Gamma(J_1)$\footnote{For $S \subseteq C \cup F$, $\Gamma(S)$ denotes the neighboring vertices  of $S$.}, such that: (a) the diameter of $J_1$ is $O(\log n/\epsilon)$, and (b) the additional clients in the boundary $|\Gamma(T_1) \setminus J_1|$ is at most $\epsilon |J_1|$. Now, we simply \emph{delete} these boundary clients and charge them to $J_1$, incurring a capacity violation of $(1+\epsilon)$. Moreover, note that in an optimal solution, \emph{all} the clients in $J_1$ \emph{must be} assigned to facilities opened in $T_1$. Using this fact, we define our first demand in the \cckp instance by $D_1 = |J_1|$ and $f_1 = |T_1|$. Repeating this process, we get a collection of $\{(J_i,T_i)\}$ which naturally defines our \cckp instance. It is then easy to show that an $\alpha$-approximation to this instance then implies an $(O(\log n/\epsilon),  \alpha(1+\epsilon))$-bicriteria algorithm for \mckc.

\medskip \noindent {\bf LP-Based Strong Decompostion.}  It is not a-priori clear how to modify the above technique to obtain better factors for the cost. To get $O(1)$-approximations, we resort to linear programming relaxations. One can write the natural LP relaxation \eqref{feq:lp1}-\eqref{feq:lp6} described in Section~\ref{fsec:prelims} -- the relaxation has $y_{ip}$ variables which denote opening a
facility with capacity $c_p$ at $i$.
Armed with a feasible solution to the LP, we prove a \emph{stronger decomposition theorem} (Theorem~\ref{fthm:decomp}): we show that we can delete a set of clients $\Cd$ which can be charged to the remaining ones, and then partition the remaining clients and facilities into {\em two} classes.
One class $\calT$  is the so-called \emph{complete neighborhood sets} of the form $\{(J_i, T_i)\}$ with $\Gamma(J_i) \subseteq T_i$ as described above --- we define our \cckp instance using these sets. The other class $\calS$  is of, what we call, {\em roundable} sets (Definition~\ref{fdef:rnding-mkc}). Roundable sets have ``enough'' $y$-mass such that installing as many  capacities as prescribed by the LP (rounded down to the nearest integer) supports the total demand incident on the set (with a $(1+\epsilon)$-factor capacity violation). Moreover, the diameter of any of these sets constructed is $\tilde{O}(1/\epsilon)$.

\medskip \noindent {\bf Technical Roadblock.}
It may seem that the above decomposition theorem implies a reduction to the \cckp problem -- for the class $\calT$ form the \cckp instance and use black-box algorithms, while the roundable sets in $\calS$ are taken care of almost by definition.
The nub of the problem lies in the {\em supply} of capacities to each of these classes.  Sure, the \cckp instance formed from $\calT$ must have a solution if the \mckc problem is feasible, {\em but only if all the $k_p$ copies of capacity $c_p$ are available to it.} However, we have already used up some of these copies to take care of the $\calS$ sets, and what we actually have available for $\calT$ is what the {\em LP prescribes.} And this can be very off (compared to the case when the \cckp instance had all the $k_p$ copies to itself). In fact, this natural  LP relaxation has bad integrality gap (Remark~\ref{frem:ig}), that is, although the LP is feasible, any assignment will violate capacities to $\Omega(n)$ factors.% \smallskip

\medskip \noindent {\bf The Supply Polyhedra.}
The above method would be fine if the supply  prescribed by the LP to the complete-neighborhood sets in $\calT$ would satisfy (or approximately satisfy) the demands of the machines in the corresponding \cckp instance. This motivates us  to define {\em supply polyhedra} for \cckp and other related problems. Informally, the supply polyhedron (Definition~\ref{fdef:supp-poly}) of a \cckp instance is supposed to capture all the vectors $(s_1,\ldots,s_n)$ such that $s_j$ copies of capacity $c_j$ can satisfy the demands of all the machines. Conversely, any vector in this polyhedron should also be a feasible (or approximately feasible) supply vector for this instance.

 If such an object $\calP$ existed, then we could strengthen our natural LP relaxation as follows. For {\em every} collection $\calT$ of complete-neighborhood sets, we add a constraint (described as \eqref{feq:lp7}) stating that the fractional capacity allocated to the facilities in $\calT$ should
 lie in the supply polyhedron of the corresponding \cckp instance. Note that this LP has exponentially many constraints, and it is not clear how to solve it. However, we can use the ``round-and-cut'' framework exploited earlier in many papers~\cite{CarrFLP00,ChakrabartyCKK11,AnSS14,DemirciL16,Li15,Li16}. Starting with a solution $(x,y)$, we use the strong decomposition theorem to obtain the set $\calT$ and check if the restriction of $y$ to the facilities in $\calT$ lies in the supply polyhedron of the corresponding \cckp instance. If yes, then we are done.
 If no, then we have obtained a separating hyperplance for the super-large LP\eqref{feq:lp1}-\eqref{feq:lp7}, and we can run the ellipsoid algorithm. In sum, we obtain an algorithm which reduces the \mckc problem to obtaining good supply polyhedra for the \cckp problem (Theorem~\ref{fthm:reduction}).% \smallskip

\medskip \noindent {\bf Supply Polyhedron for \cckp and $Q||C_{min}$.}
Do good supply polyhedra exist for \cckp or even the simpler $Q||C_{min} $ problem? Unfortunately, we show (Theorem~\ref{fthm:no-supp}) that there cannot exist {\em arbitrarily good} supply polyhedra. More precisely, there exists an instance of the $Q||C_{min}$ problem such that
for {\em any} convex set which contains all feasible supply vectors, it also contains integer supply vectors which can't satisfy all demands even when a violation of $1.001$ in capacities  is allowed. This observation exhibits the limitation of our approach: we cannot hope to obtain $(1+\epsilon)$-violation to the capacities for arbitrarily small $\epsilon$.

Nevertheless, for $Q||C_{min}$ we describe a $2$-approximate supply polyhedron (Theorem~\ref{fthm:asslp}) based on the natural assignment LP, which along with our reduction proves Theorem~\ref{fthm:2a}. In fact, we show (Lemma~\ref{flem:implied}) that for the \mckc problem with soft capacities, the strong  inequalities \eqref{feq:lp7} that we add for this $2$-approximate supply polyhedron are already implied by \eqref{feq:lp1}-\eqref{feq:lp6}.

For $Q|f_i|C_{min}$ we describe a supply polyhedron based on the {\em configuration LP} and prove that is $O(\log D)$-approximate (Theorem~\ref{fthm:conflp}) where $D$ is the ratio of maximum and minimum demand. This also implies a {\em polynomial time} $O(\log D)$-approximation algorithm for the \cckp problem. As remarked in Remark~\ref{frem:cckp-prev}, this is considerably better than any polynomial time algorithm implied before. %Our algorithm (described in Section~\ref{fsec:conflp}) proceeds by massaging a given LP solution till it yields.
%\comment{Do we want to add a line on how we get these results}
We complement this by showing (Theorem~\ref{fthm:conf-ig}, \Cref{fsec:conf-ig}) that the integrality gap of the configuration LP is $\Omega(\log D/\log\log D)$.
Moreover, our example also shows (Theorem~\ref{fthm:supply-bad-cckp}) that any convex set containing all feasible supply vectors for instances of \cckp also contains an integer supply vector which cannot satisfy all demands even when the 
capacitiy violation is allowed to be $\Theta(\log n/\log\log n)$. This settles the best algorithms we can obtain via our techniques.
On the other hand, using fairly standard tricks of enumeration and rounding, we can provide a QPTAS for \cckp (Theorem~\ref{fthm:q}). We leave the complexity of \cckp as an interesting open question.

\subsection{Related Work}\label{fsec:related}
Capacitated Location problems have a rich literature although most of the work has focused on versions where each facility arrives with a predetermined capacity and the decision process is to whether open a facility or not.
We have already mentioned the state of the art for capacitated $k$-center problems.
For the capacitated facility location problem a $5$-approximation is known via local search~\cite{BansalGG12}, while more recently an $O(1)$-approximate {\em LP-based} algorithm was proposed~\cite{AnSS14}.
All these are true approximation algorithms in that they do not violate capacities. It is an outstanding open problem to obtain true approximations for the capacitated $k$-median problem.
The best known algorithm is  the recent work of  Demirci and Li~\cite{DemirciL16} who for any $\epsilon>0$ give a $\poly(1/\epsilon)$-approximate algorithm violating the capacities by $(1+\epsilon)$-factor.
The technique of this algorithm and its precursors~\cite{AnSS14,Li15,Li16}  are similar to ours in that they follow the round-and-cut strategy to exploit exponential sized linear programming relaxations.\smallskip

The \cckp problem is a cardinality constrained max-min allocation problem. There has been some work in the scheduling literature on cardinality-constrained min-max problem.
When all the machines are identical, the problem is called the $k_i$-partitioning problem~\cite{BabelKK98}.
When the number of machines is a constant, Woeginger~\cite{Woe05} gives a FPTAS for the problem, and the best known result is a $1.5$-approximation due to Kellerer and Kotov~\cite{KellererK11}.
To our knowledge, the related speeds case has not been looked at. When the machines are unrelated,  Saha and Srinivasan~\cite{SahaS10} showed a $2$-approximation; in fact this follows from the Shmoys-Tardos rounding of the assignment LP~\cite{ShmoysT93}.\smallskip

As we have discussed above, the \mckc problem behaves rather differently than the usual homogeneous capacitated $k$-center problem. This distinction in complexity when we have heterogeneity  in resource is a curious phenomenon which deserves more attention.
A previous work~\cite{ChakrabartyGK16} of the first two authors (with P.~Goyal) looked at the (uncapacitated) $k$-center problem where the heterogeneity was in the radius of the balls covering the metric space.
As in our work, even for that problem one needs to resort to bicriteria algorithms where the two criteria are cost and {\em number} of centers opened. That paper gives an $\left(O(1),O(1)\right)$-approximation algorithm.
In contrast, we do not wish to violate the number of capacities available at all (in fact, the problem is considerably easier if we are allowed to do so -- we do not expand on this any further).

\subsection{Roadmap}
In ~\Cref{fsec:prelims}, we set up the notation and key definitions which we will subsequently use in the remaining sections. Then in~\Cref{fsec:regiongrowing}, we give our simpler weak-decomposition theorem which (upto a logarithmic factor in the distance objective) effectively reduces \mckc to \cckp. To overcome this logarithmic loss in the distance objective, we turn to an LP-based approach and a stronger decomposition theorem. But to help us along the way, we introduce and state our main results about the so-called \emph{supply polyhedra} for \cckp in~\Cref{fsec:supplypolyhedra}. In~\Cref{fsec:o1} we then state our strong decomposition theorem and show how it can be combined with good supply polyhedra to get~\Cref{fthm:2,fthm:2a}. In the next~\Cref{fsec:decomp-proof}, we prove the strong decomposition theorem. Subsequently, in~\Cref{fsec:asslp,fsec:conflp}, we prove the existence of good supply polyhedra for $Q || C_{min}$ and \cckp. In Section~\ref{fsec:conf-ig} we provide integrality gap examples.
Finally in~\Cref{fsec:qptas} we show that \cckp admits a QPTAS, thereby proving~\Cref{fthm:q}.

%\newpage
\section{Technical Preliminaries}\label{fsec:prelims}
Given an \mckc instance, we start by guessing $\opt$. We either prove $\opt$ is infeasible, or find an $(a,b)$-approximate allocation of clients to facilities.
	We define the bipartite graph $G = (F\cup C,E)$ where $(i,j)\in E$ iff $d(i,j) \leq \opt$. If $\opt$ is feasible, then the following assignment LP\eqref{feq:lp1}-\eqref{feq:lp6}
	must have a feasible solution.
In this LP, we  have opening  variables $y_{ip}$ for every $i\in F,p\in [P]$ indicating whether we open a facility with capacity $c_p$ at location $i$. Recall that the capacities available to us are $c_1, c_2, \ldots, c_P$ -- a facility with
capacity $c_p$ installed on it will be referred to as a {\em type $p$ facility.}
	We have connection variables $x_{ijp}$ indicating the fraction to which client $j\in C$ connects to a facility at location $i$ where a type $p$ facility has been opened.
	We force $x_{ijp} = 0$ for all pairs $i,j$ and type $p$ such that  $d(i,j) > \opt$.

		\begin{minipage}{0.45\textwidth}
			\begin{alignat}{4}
				& \quad \forall j\in C,   &&\quad  \textstyle \sum_{i\in F} \sum_{p\in [P]}  x_{ijp} \geq 1 \label{feq:lp1} \tag{\small{L1}}  \\
				& \quad \forall i\in F,p\in [P] ,  &&\quad  \textstyle \sum_{j\in C}  x_{ijp} \leq c_py_{ip} \label{feq:lp2} \tag{\small{L2}} \\
				& \quad \forall p\in [P], && \quad \textstyle \sum_{i\in F} y_{iq}   \leq k_p \label{feq:lp3}  \tag{\small{L3}}
			\end{alignat}
		\end{minipage}
		~\vline~
		\begin{minipage}{0.45\textwidth}
			\begin{alignat}{4}
				& \quad \forall i\in F, j\in C,p\in [P],  && \quad x_{ijp} \leq y_{ip}\label{feq:lp4}   \tag{\small{L4}} \\
				& \quad \forall i\in F, && \quad \textstyle\sum_{p\in [P]} y_{ip} \leq 1 \label{feq:lp5}  \tag{\small{L5}} \\
				& \quad \forall i\in F,j\in C,p\in [P], && \quad x_{ijp},y_{ip} \geq 0\label{feq:lp6}\tag{\small{L6}}
			\end{alignat}
		\end{minipage}
\smallskip

\noindent		
We say a solution $(x,y)$ is $(a,b)$-feasible if it satisfies \eqref{feq:lp1}, \eqref{feq:lp3}-\eqref{feq:lp6}, and \eqref{feq:lp2} with the RHS replaced by $bc_py^\mathsf{int}_{ip}$, and $x_{ijp} > 0$ only if $d(i,j) \leq a\cdot \opt$,
We desire to find an integral solution $(x^\mathsf{int},y^\mathsf{int})$ which is $(a,b)$-feasible.
The following lemma shows that it suffices just to round the $y$-variables.
\begin{claim}
Given an $(a,b)$-feasible solution $(x,\y)$ where $\y_{ip}\in \{0,1\}$,
we can get  an $(a,b)$-approximate solution to the \mckc problem.
\end{claim}
\begin{proof}
Consider a bipartite graph with client nodes $C$ on one side, and nodes of the form $(i,p)$ with $\y_{ip} = 1$ on the other. The node $(i,p)$ has capacity $bc_p$.
Since $(x,\y)$ satisfies the conditions of the lemma, there is a fractional matching in this graph so that each client $j$  is fractionally matched to an $(i,p)$ so that $d(i,j)\leq a\cdot \opt$,
and the total fractional load on $(i,p)$ is $\leq bc_p$. The theory of matching tells us that there is an {\em integral} assignment of clients $j$ to nodes $(i,p)$ such that $d(i,j)\leq a\cdot\opt$
and the number of nodes matched to $(i,p)$ is $\leq \ceil{bc_p}$. Therefore opening a capacity $c_p$ facility at $i$ for all $(i,p)$ with $\y_{ip} = 1$ gives an $(a,b)$-approximate solution to \mckc.
\end{proof}
\noindent
Henceforth, we focus on rounding the $y$-values. To this end, we make the following useful definition.
\begin{definition}[Roundable Sets]\label{fdef:rnding-mkc}
	A set of facilities $S\subseteq F$ is said to be $(a,b)$-roundable w.r.t $(x,y)$ if
	\begin{itemize}[noitemsep]
		\item[(a)] $\diam_G(S) \leq a$
		\item[(b)] there exists a rounding $\y_{ip} \in \{0,1\}$ for all $i \in S, p\in [P]$ such that
		\begin{enumerate}
			\item $\sum_{q \geq p} \sum_{i\in S} \y_{iq} ~\leq~ \floor{\sum_{q \geq p}\sum_{i\in S} y_{iq}}$ for all $p$, and
			\item $\sum_{j\in C} d_j \sum_{i\in S,p\in [P]} x_{ijp} \leq b\cdot \sum_{i\in S} \sum_{p\in [P]} c_p \y_{ip}$
		\end{enumerate}
	\end{itemize}
\end{definition}
\noindent
If $(x,y)$ were feasible, then for any $(a,b)$-roundable set, we can integrally open facilities to satisfy all the demand that was fractionally assigned to it taking a hit of $a$ in the cost and a factor of $b$ in the capacities. Furthermore, the number of open facilities is at most what the LP prescribes. Therefore, if we would be able to decompose the instance into roundable sets, we would be done.
Unfortunately, that is not possible, and in fact the above LP has a large integrality gap even when we allow arbitrary violation of capacities.

\begin{remark}[Integrality Gap for \mckc] \label{frem:ig}
Consider the following instance. The metric space $X$ is partitioned into $(F_1\cup C_1) \cup \cdots \cup (F_K\cup C_K)$, with $|F_k| = 2$ and $|C_k| = K$ for all $1\le k\le K$.
The distance between any two points in $F_i\cup C_i$ is $1$ for all $i$, while all other distances are $\infty$. The capacities available are $k_1 = K$ facilities with capacity $c_1 = 1$ and
$k_2= K-1$ facilities with capacity $c_2 = K$. It is easy to see that integrally any solution would violate capacities by a factor of $K/2$.
%It is easy to see that the above instance is not feasible with $OPT=1$: indeed, there is at least one client location where the optimal solution does not place a facility of capacity $H$ in its neighborhood, and it is not possible to serve the demand of this client using only capacity $1$ facilities, as there are only two locations where we can place facilities in its neighborhood.
On the other hand, there is a feasible solution for the above LP relaxation: for $F_k = \{a_k,b_k\}$, we set $y_{a_k2} = 1-1/K$ and $y_{b_k1} = 1$, and for all $j\in C_k$, we set $x_{a_kj2} = 1-1/K$ and $x_{b_kj1} = 1/K$.

 For the version with soft capacities, we do not have the constraint \eqref{feq:lp5} and the above integrality gap doesn't hold since we can install capacity $K$ facilities on $K-1$ of the sets $F_k$'s, $1\leq k\leq K-1$, and $K$ copies of the capacity $1$ facilities at $F_K$. Note that although $|F_K| = 2$, we have opened $K$ capacities.
\end{remark}

In particular, note that for the $(x,y)$ solution in the integrality gap example above there are no roundable sets. This motivates the definition of the second kind of sets.

\begin{definition}[Complete Neighborhood Sets] \label{fdef:comp-nbr}
	A subset $T\subseteq F$ of facilities is called a {\em complete neighborhood} if there exists a client-set $J\subseteq C$ such that $\Gamma(J) \subseteq T$.
	In this case the subset $J$ is said to be {\em responsible} for $T$. Additionally, a complete neighborhood $T$ is said to be an $\alpha$-complete neighborhood if $\diam(T) \leq \alpha$.
\end{definition}
\begin{remark}[Complete Neighborhood Sets to \cckp]\label{frem:red}
	\emph{
If we find a complete neighborhood $T$ of facilities with say a set $J$ of clients responsible for it, then we know that the optimal solution must satisfy all the demand in $J$ by suitably opening facilities of sufficient capacity in $S$. Given a collection $\calT = (T_1,\ldots,T_m)$ of disjoint $\alpha$-complete neighborhood sets with $J_i$ repsonsible for $T_i$, we can define an instance $\calI$ of the \cckp problem with $m$ machines with demands $D_i = |T_i|$ and cardinality constraint $f_i = |T_i|$, and $P$ jobs of capacities $c_1,\ldots,c_P$. The facilities opened by the $\opt $ solution corresponds to a valid solution for $\calI$; furthermore, any $\beta$-approximate solution for $\calI$ corresponds to a $(\alpha,\beta)$-approximate solution for the \mckc problem restricted to clients in $\cup_{\ell} J_\ell$. Finally note that for \mckc with soft-capacities, $\calI$ is an instance of the $Q||C_{min}$ problem.
}

\emph{
Note that the above integrality gap  example is essentially a \cckp instance with $K$ machines of demand $K$ each having cardinality constraint $2$, and there are $K$ jobs of capacity $1$ and $K-1$ jobs with capacity $K$. This shows the assignment LP has bad integrality gap for the \cckp problem (but not for $Q||C_{min}$).}
\end{remark}

Our final definition is that of $(\tau,\rho)$-{\em deletable clients} who can be removed from the instance since they can be ``$\rho$-charged'' to the remaining clients no further than $\tau$-away.
\begin{definition}[Deletable Clients]\label{fdef:deletable}
	A subset $\Cd\subseteq C$ of clients is $\rho$-deletable if there exists a mapping $\phi_{j,j'}\in [0,1]$ for $j\in \Cd$ and $j'\in C\setminus \Cd$ satisfying (a) $\sum_{j'\in C\setminus \Cd} \phi_{j,j'} = 1$ for all $j\in \Cd$, and(b) $\sum_{j\in \Cd} \phi_{j,j'} \leq \rho$ for all $j'\in C\setminus \Cd$. Furthermore, $\phi_{j,j'} > 0$ only if $d(j,j') \leq \tau\cdot\opt$.
\end{definition}
 The following claim shows we can remove $\Cd$ from consideration.
\begin{claim}\label{fclm:prelim3}
	Let $\Cd$ be a $(\rho,\tau)$-deletable set.
	Given an $(a,b)$-approximate feasible solution $(x',\y)$ where $x'_{ijp}$ is defined only for $j\in C\setminus \Cd$, we can extend $x'$ to a general $(x,\y)$ solution
	which is $(a+\tau, b(1+\rho))$-approximate feasible.
\end{claim}
\begin{proof}
For any $j\in \Cd$, define $x_{ijp} = \sum_{j'\in C\setminus \Cd} x_{ij'p}\phi_{j,j'}$.
We get for all $j\in \Cd$,
$\textstyle \sum_{i\in F} \sum_{p\in [P]} x_{ijp} = \sum_{i,p} \sum_{j'\in C\setminus \Cd} x_{ij'p}\phi_{j,j'} = \sum_{j'\in C\setminus \Cd} \phi_{j,j'} \left(\sum_{i,p} x_{ij'p}\right) \geq \sum_{j'\in C\setminus \Cd} \phi_{j,j'} = 1$,
and for all $i\in F,p\in [P]$,
$\textstyle \sum_{j\in \Cd}  x_{ijp} = \sum_{j\in \Cd} \sum_{j'\in C\setminus \Cd} x_{ij'p}\phi_{j,j'} = \sum_{j'\in C\setminus \Cd} x_{ij'p}\left( \sum_{j\in \Cd} \phi_{j,j'}\right)  \leq \rho \sum_{j'\in C\setminus \Cd} x_{ijp} \leq b\rho c_p$. Therefore, in all we have $\sum_{j\in C} x_{ijp} \leq bc_p(1+\rho)$.
\end{proof}

% described in Idea 1.

%\newpage
\section{Reduction to Max-Min Allocation via Region Growing}\label{fsec:regiongrowing}
In this section, we give a reduction to \cckp  when we allow logarithmic approximations. We then show how we get Theorem~\ref{fthm:1} using this result.
\begin{theorem}\label{fthm:weakred}
	Given an $\beta$-approximation algorithm for \cckp (respectively, $Q||C_{min}$), for any $\epsilon>0$ there exists an $\left(O(\log n/\epsilon), \beta(1+\epsilon)\right)$-approximate algorithm for the \mckc problem (respectively, for the \mckc problem with soft capacities).
\end{theorem}

The main crux of the above proof is the following decomposition theorem obtained by the technique of region growing which was first used in the context of sparsest and multi cut problems~\cite{LeightonR99,GargVY96}.

\begin{theorem}\label{fthm:weakdecomp}
	Given a guess $\opt$ for \mckc problem and any $\epsilon>0$,  there is an algorithm which partitions the facilities $F$ into a collection $\calT = (T_1,\ldots,T_L)$ of
	$O(\log n/\epsilon)$-complete neighborhoobd sets with $J_\ell$ responsible for $T_\ell$, and the client set $C = \Cd \cup \bigcup_{\ell=1}^L J_\ell$ such that
	$\Cd$ is an $(O(\log n/\epsilon),\epsilon)$-deletable set.
\end{theorem}
\begin{proof}
Recall $G$ is the graph with $d(i,j) \leq \opt$ for $(i,j)\in G$. %	We start with a guess of $\opt$ and the $\opt$-induced graph $G$. We assume this is a correct guess.
%	Given $G$ we given an algorithm construct an instance $\calI$ of $Q|k_i|C_{min}$.
	Initially $\calT$ and $\Cd$ are empty. We maintain a set of alive clients $C'$ which is initially $C$. We maintain a working graph $H$ which is initialized to $G$ and is always a subgraph of $G$.
	Given a node $j$ and an integer $t$, let $N^{(t)}_H(j)$ denote all the modes $j'$ s.t. $d_H(j,j') < t$ and $\Gamma^{(t)}_H(j)$ denote all the nodes $j'$ with $d_H(j,j')  = t$.
	Note that for even $t$, we have $\Gamma^{(t)}_H(j) \subseteq C$, and for odd $t$, $\Gamma^{(t)}_H(j) \subseteq F$.
	\smallskip

	Till $C'$ is empty, we perform the following operation.
	Select an arbitrary active client $j\in C'$. Find the smallest even $t$ such that $|\Gamma^{(t)}_H(j)| < \eps\cdot |N^{(t)}_H(j)\cap C|$. Since for all $s < t$ we have $|N^{(s+2)}(j)\cap C|\geq (1+\epsilon)|N^{(s)}_H(j)\cap C|$,  $|N^{(s+2)}(j)\cap C| > (1+\epsilon)^{\frac{s}{2}}$. Therefore,  $t \leq (2\ln n)/\epsilon,$ where $n=|C'|$. %Let us use $A_j$ to denote $N_t(j)$ and $B_j$ to denote $\Gamma_{t}(j)$.
	We define $T_\ell := N^{(t)}_H(j) \cap F$ and $J_\ell := N^{(t)}_H \cap C$; note that $T_\ell$ is an $O(\log n/\epsilon)$-complete neighborhood which is responsible for $J_\ell$. Furthermore, we add $J_\mathsf{ext} := \Gamma^{(t)}_H(j)$ to $\Cd$, and since
	$|J_\mathsf{ext}| < \eps |J_\ell|$ and $\diam(J_\ell) = O(\log n/\epsilon)$, there exists a mapping $\phi_{j,j'}$ for $j\in J_\mathsf{ext}$ and $j'\in J_\ell$ such that $\sum_{j'\in J_\ell} \phi_{j,j'} = 1$ for all $j\in J_\mathsf{ext}$, and
	$\sum_{j\in J_\mathsf{ext}} \phi_{j,j'} \leq \epsilon$ for all $j\in J_\ell$, and $\phi_{j,j'} > 0$ only if $d(j,j') = O(\log n/\epsilon)$. That is, $J_\mathsf{ext}$ is a valid $(O(\log n/\epsilon),\epsilon)$-deletable set.
	Finally, we delete $T_\ell \cup J_\ell \cup J_\mathsf{ext}$ from $H$ and $J_\ell \cup J_\mathsf{ext}$ from $C'$. We continue this procedure till $C'$ is empty. \end{proof}
\begin{proof}[Proof of Theorem~\ref{fthm:weakred}]
	Given $\calT$ we form the instance $\calI$  of \cckp (or $Q||C_{min}$ in case of soft-capacities) described in Remark~\ref{frem:red}. We provide $k_p$ copies of job with capacity $c_p$. If $\opt$ is feasible, then there must exist a feasible solution to $\calI$. Furthermore, a $\beta$-approximate solution to $\calI$ gives an $\left(O(\log n/\epsilon),\beta\right)$-approximate solution to the clients in $C\setminus \Cd$. The theorem follows from Claim~\ref{fclm:prelim3}.
%	
%	\begin{claim}
%		There exists an allocation where the total processing time of each machine is at least $1$.
%	\end{claim}
%	\begin{proof}
%		In the optimal solution to \mckc, a total capacity of $|A_j|$ must be installed among the facilities $F_j$.
%		We mimic this solution for $\calI$, and since the speed of machine $j$ is $1/|A_j|$, the claim follows.
%	\end{proof}
%	\begin{claim}
%		Given an allocation to $\calI$ where each machine gets processing time $\geq \alpha$, we can construct a solution for the \mckc problem
%		which is $O(\log n/\epsilon)$-approximate and violates the capacities by at most $(1+\epsilon)/\alpha$-factor.
%	\end{claim}
%	\begin{proof}
%		For machine $j$ in $\calI$, let $S_j$ be the jobs allocated; we have $\sum_{p\in S_j} c_p \ge \alpha\cdot |A_j|$ and $|S_j| \leq k_j = |F_j|$.
%		We arbitrarily open $|S_j|$ facilities in $F_j$ to which we assign all the clients in $A_j \cup B_j$. Since $|B_j|<\eps |A_j|$, we have the total capacity opened is at least $\alpha/(1+\epsilon)$ times
%		the total number of clients in $A_j\cup B_j$. Furthermore, every client in $A_j\cup B_j$ is at most $O(\log n/\epsilon)$-away from any facility in $F_j$. This implies the assignment.
%	\end{proof}	
\end{proof}

As a corollary to Theorem~\ref{fthm:weakred}, and using the fact that $Q||C_{min}$ has a PTAS ~\cite{AzarE98},  and our result (Theorem~\ref{fthm:q} proved in~\Cref{fsec:qptas}) that \cckp has a quasipolynomial time approximation scheme (QPTAS), we get Theorem~\ref{fthm:1}.\medskip

In Section~\ref{fsec:o1} we state a much stronger decomposition theorem than Theorem~\ref{fthm:weakdecomp} which exploits the LP solution.
To exploit it for \mckc problem, however, and prove an analogous theorem as Theorem~\ref{fthm:weakred}, we need to understand certain polyhedra with respect to the \cckp problem.
We first do this in the next section.

%\newpage
\section{Max-Min Allocation Problems and Supply Polyhedra} \label{fsec:supplypolyhedra}
An instance of the \cckp problem has $m$ machines $M$ with demands $D_1,\ldots,D_m$ and cardinality constraints $f_1,\ldots, f_m$, and $n$  types of jobs $J$ with capacities $c_1,\ldots,c_n$ respectively.
In $Q||c_{min}$, there are no $f_i$'s, or equivalently $f_i = \infty$.
%In the version with cardinality constraints, that is $Q|f_i|C_{min}$ we are also given positive integers $f_1,\ldots, f_m$.

A {\em supply vector} $(s_1,\ldots,s_n)$ where each $s_j$ is a non-negative integer
is called {\em feasible} for instances of these problems if the ensemble formed by $s_j$ copies of jobs of capacity $c_j$ can be allocated feasibly to satisfy all the demands.
The {\em supply polyhedra} of these instances desires to capture these feasible supply vectors.

\begin{definition}[Supply Polyhedron]\label{fdef:supp-poly}
	Given an instance $\calI$ for a max-min allocation problem, a polyhedron $\calP(\calI)$ is called an $\alpha$-approximate supply polyhedron if
	(a) all feasible supply vectors lie in $\calP(\calI)$, and (b) given any non-negative integer vector $(s_1,\ldots,s_n)\in \calP(\calI)$ there exists an assignment
	of the $s_j$ jobs of capacity $c_j$ to the machines such that machine $i$ receives a total capacity of $\geq D_i/\alpha$.
\end{definition}

Ideally, we would like {\em exactly} supply polyhedra. One guess would be the convex hull of all the supply vectors; indeed this is the tightest polytope satisfying condition (a).
Unfortunately, there are instances of $Q||C_{min}$ (and even for the uniform case $P||C_{min}$) where the convex hull of supply vectors contains infeasible integer points.
This rules out exact or even $(1+\epsilon)$-approximate supply polyhedra. In~\Cref{fthm:supply-bad-cckp} in~\Cref{fsec:conf-ig}, we show a stronger lower bound of $\Omega(\log D/\log \log D)$ on the best approximation-factor of any supply polyhedra for \cckp.
%This instance is motivated by integrality gap examples for machine scheduling~\cite{bibid}.
\begin{theorem}\label{fthm:no-supp}
	There cannot exist $\alpha$-approximate supply polyhedra (or convex sets) for $\alpha < 1.001$ for all  $P||C_{min}$ instances.
\end{theorem}
\begin{proof}
%\comment{\bf \Large	Needs to be written }.
The example is almost similar to the example in \cite{KurpiszMMMVW16} which was used to show integrality gap examples for strong LP relaxations for identical machines makespan minimization problem.
We just sketch a proof here. Recall the Petersen Graph with $10$ nodes and $15$ edges which has the following key property: it has six perfect matchings $M_1,\ldots,M_6$ such that each edge $(i,j)$ appears in exactly $2$ of these matchings; however, its edge set cannot be partitioned into $3$ perfect matchings.
The vertices are numbered $0,1,\ldots,9$.

Now we can describe the instance. Fix $k$ to be any positive integer.
We have $15$ types of jobs $p_{ij} = 2^i + 2^j$ for every edge $(i,j)$ of the Petersen graph.
We have $3k$ machines each with the same demand $D = \sum_{i=0}^9 2^i = 1023$.
Consider the six supply vectors $s^{(t)}$ for $1\leq t\leq 6$,  which contains $3k$ copies of the job corresponding to edge $(i,j)$ iff $(i,j)$ is in the matching $M_t$.
These are feasible supply vectors; indeed assign each of the $3k$ machines one jobs $p_{ij}$ for $(i,j) \in M_t$. Now any convex set (in particular polyhedra) containing these six supply vectors
must contain any convex combination. However the vector $\frac{1}{6}\sum_{t=1}^6 s^{(t)}$ is an integer vector with $k$ copies of each $(i,j)$ for all edges of the Petersen Graph.
This uses the fact that every edge is in exactly two perfect matchings. Since the edges of the  Petersen graph can't be partitioned into $3$ perfect matchings, any allocation of this supply vector
must give one machine demand $\leq 1022$. Therefore, there can't be any $\alpha$-approximate supply polyhedra for $P||C_{min}$.
\end{proof}
\begin{remark}\emph{
At this point, we should underscore the difference between supply polyhedra and say LP relaxations for solving  these allocation problems.
Given an instance of say $Q||C_{min}$ {\em along with} the supply vector (which is one standard way the problems are stated), there does exist a polytope capturing all the feasible allocations. It is the integer hull.
However, in general, the description of this integer hull uses the supply vector in describing these constraints and therefore are non-linear when the supplies are variables. Nevertheless, as we discuss below, many LP relaxations
studied in the literature imply supply polyhedra, and their integrality gaps imply the approximation factor for the polyhedra as well.
}
\end{remark}

For our purposes, we need more technical conditions from the supply polyhedra. The first is a natural condition which states that if one moves the supply to higher capacity jobs, then feasibility remains.
The second is related to polynomial time algorithms.

\begin{definition}
	A supply polyhedron $\calP(\calI)$ is {\em upward-feasible} if the following condition holds.
Reorder the jobs so that $c_1\le c_2 \le \cdots \le c_n$.
		If $(s_1,\ldots,s_n)\in \calP$ and $(t_1,\ldots,t_n)$ is a non-negative vector satisfying $t\succeq_\suff s$, that is, $\sum_{k\geq i} t_k \geq \sum_{k\geq i} s_k$, then $(t_1,\ldots,t_n)\in \calP$ as well.
\end{definition}
\begin{definition}[$\gamma$-Approximate Separation.]
A $\gamma$-approximate separation oracle for the supply polyhedron $\calP(\calI)$ is a polynomial time procedure which 	given any $y\in \R^n_{\geq 0}$,  either returns a hyperplane separating $y$ from $\calP$, or asserts that
		$y\in \calP(\calI')$ for the supply polyhedra of the instance $\calI'$ where all demands have been reduced by a factor $\gamma$.
%	\end{asparaitem}
\end{definition}
%Finally, note that if given
\subsection{Approximate Supply Polyhedra for $Q||C_{min}$}
\def\pv{\mathbf{b}}
%
%Let the instance $\calI$ of $Q||C_{min}$  have $m$ machines $M$ with demands $D_1 \geq \cdots \geq D_m$ and $n$ types of  jobs $J$ with capacities $c_1 \geq \cdots \geq c_n$.
%A supply vector $(s_1,\ldots,s_n)$ indicates the number of jobs of each type available; a supply vector is feasible if together they can satisfy all the demands.
%We wish to find a convex set/polyhedra which captures all the feasible supply vectors. In particular, any feasible supply vector should be in the set, and given any (integer) supply vector in the set
%there should be an allocation which satisfies the demands to an $\alpha$-factor.
%\smallskip
%
\noindent
For $Q||C_{min}$, the following assignment LP acts as a good supply polyhedra.
%A feasible supply vector $(s_1,\ldots,s_n)$ must lie in the following polytope.
%We assume we have a guess $D$ for the optimum value which is certified by a feasible solution to the following assignment LP. Below, $D_i := Ds_i$.
\begin{alignat}{4}
	\calP_\mathsf{ass}(\calI)  && = \{(s_1,\ldots,s_n):  && \notag \\
	&& \quad \forall j \in J,   &\quad  \textstyle \sum_{i\in M} z_{ij}  \leq  s_j \label{feq:asslp1} \tag{A1} \\
	&& \quad \forall i\in M ,  &\quad  \textstyle \sum_{j\in J}  z_{ij}  \min(c_j,D_i) \geq D_i \label{feq:asslp2} \tag{A2}\\
	&& \quad \forall i\in M, j\in J, & \quad z_{ij}   \geq 0 \label{feq:asslp3}  \} \tag{A3}
\end{alignat}
\begin{theorem}\label{fthm:asslp}
	For any instance $\calI$ of $Q||C_{min}$, $\calP_\mathsf{ass}(\calI)$ is an upward feasible, $2$-approximate supply polyhedron with exact separation oracle.
%	Given $(s_1,\ldots,s_n) \in \calP_\mathsf{ass}$, there is an of assignment $\phi$ of the $s_j$ jobs of capacity $c_j$  to the machines such that for all $i\in M$,
%	$\sum_{j:\phi(j) = i} c_j \geq D_i/2$.
\end{theorem}
\noindent
We defer the proof of the above theorem to Section~\ref{fsec:asslp}.

\subsection{Approximate Supply Polyhedra for $Q|f_i|C_{min}$}
For \cckp, a candidate supply polyhedra would be \eqref{feq:asslp1}-\eqref{feq:asslp3} along with
\begin{equation}
\textstyle \forall i\in M, \qquad \sum_{j\in J} z_{ij} ~\leq~ f_i \tag{A4} \label{feq:asslp4}
\end{equation}
which would enforce the cardinality constraint. Unfortunately, an example akin to that in Remark~\ref{frem:ig} shows that $\calP_\mathsf{ass}$ is not
an $\alpha$-approximate supply polyhedron for \cckp instances with $\alpha = o(n)$. We define a stronger supply polyhedron.
However, at this juncture we state a theorem regarding \eqref{feq:asslp1}-\eqref{feq:asslp4} which
is based on the ideas from Shmoys-Tardos rounding~\cite{ShmoysT93}.
%we will use at least twice.
\begin{theorem}[Shmoys-Tardos Rounding]
	\label{fthm:shmoystardos}
	Let \eqref{feq:asslp1}-\eqref{feq:asslp4} have a feasible solution $z$ and let $C_i := \max_{j:z_{ij}>0} c_j$.
There is an integral assignment $\zz_{ij}\in \{0,1\}$ which satisfies \eqref{feq:asslp1}, \eqref{feq:asslp4}, and \eqref{feq:asslp2} with the RHS replaced by $D_i - C_i$ for all $i$.
\end{theorem}
	\begin{proof} (Sketch) We proceed as in  the Shmoys-Tardos rounding~\cite{ShmoysT93} of the assignment LP.
	We convert the instance into a bipartite matching instance where on one side we have the jobs $J$ with
		multiplicities $s_1,\ldots,s_n$, and on the other side we have the machines where we take $f'_i := \ceil{\sum_{j\in J} z_{ij}} \le f_i$ copies of each machine.
		The solution $z$ is converted to a (fractional) solution $\z$ on this bipartite graph where each job $j$
		is assigned  by $\z$ to an extent of at most $1$. Furthermore, for every machine $i$, each of its $f'_i$ copies, except perhaps for the last one,
		 gets $\z$-mass exactly $1$.  This assignment also has the property that for any machine $i$ and job $j$,
		if  $\z$ assigns (fractionally) $j$ to $\ell^{th}$ copy of $i$, then $c_j$ is at least the total fractional demand assigned by $z$ to
		the $(\ell+1)^{th}$ copy of this machine.
		 %and for any $\ell$,
		 %the $\ell$th copy any job $j$ which gives $\z$-mass to this has $c_j$ more than the total fractional contribution to the $\ell+1$th copy.
		Since each copy of a machine (except for the last copy) gets $\z$-mass exactly $1$, there is an assignment of jobs to these copies such that each such copy of machine $i$ gets exactly one job. We give machine $i$ whatever its copies obtain; note that it obtains $f'_i \leq f_i$ jobs.
		The total capacity of jobs allocated is therefore $\ge D_i - \Delta$ where $\Delta$ is the fractional capacity assigned to $i$'s first copy. Since $\Delta\leq C_1$,  this proves the theorem.
%		
%		{\large NEEDS BETTER WRITING} We repeat the argument of Shmoys and Tardos~\cite{bibid}.
%	Form $\floor{\sum_{j\in J} z_{ij}} \le f_i$ copies of every machine; let $N_i$ be the copies of machine $i$. Order the jobs with multiplicities s.t. $c_1 \geq c_2 \geq \cdots \geq c_N$ where $N = \sum_j s_j$.
%	Modify $z_{ij}$ to get an assigment $z_{ij}$ for $i\in \cup N_i$ and $j\in [N]$ as follows. We do this for one machine $i$.
%	
%	Given $z_{ij}$'s we form $|N_i| + 1$ groups $S_1,\ldots, S_{|N_i|},S_{|N_i|+1}$ with $\sum_{j\in S_t} z_{ij} =1$ for all $1\leq t\leq |N_i|$ and $\sum_{j\in S_t} z_{ij} < 1$ for $t = |N_i|+1$.
%	Note that $\sum_{j\in S_t} z_{ij}c_j < c_{j'}$ for $j'\in S_{t-1}$. When we do this modification for all machines, we get a fractional matching solution where all the $N_i$ copies get fractional value $1$ but the jobs are at most $1$.
%	So, there is an integral matching. The total integral load on machine $i$ is at least $\sum_{t>1} \sum_{j\in S_t} z_{ij}c_j \geq D_i - C_i$ since $z_{ij} = 0$ for $c_j > C_i$.
%		
%		Cardinality constraint vacuously satisfied.
	\end{proof}

\noindent
In other words, for instances where $c_j$'s are $\ll D_i$'s, $\calP_\mathsf{ass}$ is a good supply polyhedron. But in general we need a supply polyhedra with stronger constraints. \smallskip

Let $\Supp$ be a multiset indicating infinitely many copies of jobs in $J$.
For every machine $i$, let $\calF_i :=\{S\in \Supp: |S| \leq f_i ~\textrm{ and } \sum_{j\in S} c_j \geq D_i\}$ denote all the feasible sets that can satisfy machine $i$.
Let $n(S,j)$ denote the number of copies of job of type $j$.
\begin{alignat}{4}
	\calP_\mathsf{conf}(\calI) && = \{(s_1,\ldots,s_n):  && \notag \\
	&& \quad \forall i \in M,   &\quad  \textstyle \sum_{S} z(i,S)  =  1 \label{feq:conflp1} \tag{C1} \\
	&& \quad \forall j\in J ,  &\quad  \textstyle \sum_{i\in M,S}  z(i,S)n(S,j) \le  s_j \label{feq:conflp2}\tag{C2} \\
	&& \quad \forall i\in M, S\notin \calF_i, & \quad z(i,S)  = 0 \label{feq:conflp3} \} \tag{C3}
\end{alignat}
\begin{theorem}\label{fthm:conflp}
	For any instance $\calI$ of \cckp, $\calP_\mathsf{conf}(\calI)$ is an upward feasible, $O(\log D)$-approximate supply polyhedron with $(1+\epsilon)$-approximate separation oracle for any $\eps > 0$,
	where $D := D_{\mathrm{max}}/D_\mathrm{min}$.
%	Given $(s_1,\ldots,s_n) \in \calP_\mathsf{conf}$ for an instance $\calI$ of $Q|f_i|C_{min}$, there is an of assignment of the $s_j$ jobs of capacity $c_j$  to the machines such that for all $i\in M$
%	receives a total capapcity $\geq D_i/\alpha$ for $\alpha = O(\log D)$ where $D = D_{max}/D_{min}$.
\end{theorem}
\noindent
As a corollary, we get  \Cref{fthm:cckp}.
We complement this with an almost matching integrality gap.
\begin{theorem}
	\label{fthm:conf-ig}
	The integrality gap of $\calP_\mathsf{conf}$ is $\Omega\left(\frac{\log n}{\log\log n}\right)$. More precisely, there exists an instance $\calI$ of \cckp and a supply vector $(s_1,\ldots,s_n) \in \calP_\mathsf{conf}(\calI)$, but in any feasible allocation of $s_j$ jobs of capacity $c_j$ to the machines, there exists some machine  $i$ receiving $\leq O\left( D_i\frac{\log\log n}{\log n}\right)$.
\end{theorem}
We defer the proofs of \Cref{fthm:conflp} and \Cref{fthm:conf-ig} to Section~\ref{fsec:conflp}.

%\newpage
\section{\mckc via Supply Polyhedra}\label{fsec:o1}
\def\yy{y^\calT}
In this section, we prove the following theorem. One of the main engines will be a strong decomposition theorem (\Cref{fthm:decomp}) which we will state here but will prove in the next section.
\begin{theorem}\label{fthm:reduction}
Suppose there exists $\beta$-approximate, upward feasible supply polyhedra  for all instances of $Q|f_i|C_{min}$ (respectively, $Q||C_{min}$) which have $\gamma$-approximate separation oracles.
Then for any $\delta\in(0,1)$, there is an $\left(\tilde{O}(1/\delta),\gamma\beta(1+5\delta)\right)$-bicriteria approximation algorithm for the \mckc problem (respectively, with soft capacities).
\end{theorem}
\noindent
	
\noindent
\noindent
The above theorem and results about supply polyhedra imply the bicriteria algorithms for the \mckc problem.
Theorem~\ref{fthm:2} follows from the above theorem (instantialted with $\delta = 0.5$, say) and Theorem~\ref{fthm:conflp} after noting that $D_\mathrm{max}/D_\mathsf{min} \leq n$ in the reduction we describe below. Theorem~\ref{fthm:2a} follows
from the above theorem and Theorem~\ref{fthm:asslp}.

Before moving to the proof of Theorem~\ref{fthm:reduction}, we state our main technical result which is a decomposition theorem which essentially states that given an \mckc instance, we can partition the problem into roundable and complete neighborhood sets. The reader may want to recall the definitions of roundable sets (Definition~\ref{fdef:rnding-mkc}), complete neighborhood sets (Definition~\ref{fdef:comp-nbr}), deletable sets (Definition~\ref{fdef:deletable}), and the natural LP relaxation \eqref{feq:lp1}-\eqref{feq:lp6}.
%Our main techincal theorem is the following
It is perhaps instructive to compare the below theorem with Theorem~\ref{fthm:weakdecomp}.
The proof of this theorem is rather technical, and we defer it to the next section.
\begin{theorem}[{\bf Decomposition Theorem}]\label{fthm:decomp}
	Given a feasible solution $(x,y)$ to LP\eqref{feq:lp1}-\eqref{feq:lp6}, and $\delta > 0$, there is a polynomial time algorithm which finds a solution $\x$ satisfying \eqref{feq:lp2} and\eqref{feq:lp4}, and a
	decomposition as follows.
	\begin{enumerate}%[noitemsep]
		\item The facility set $F$ is partitioned into two families $\calS = (S_1, S_2, \ldots, S_K)$ and $\calT = (T_1, T_2, \ldots, T_L)$ of mutually disjoint subsets.
		The client set $C$ is partitioned into three disjoint subsets $C = \Cd \cup \Cbb \cup \Cb$ where $\Cd$ is a $(\tilde{O}(1/\delta),\delta)$-deletable subset.
		
			\item Each $S_k \in \calS$ is $(\tilde{O}(1/\delta),(1+\delta))$-roundable with respect to $(\x,y)$, and moreover, each client in $\Cb$ satisfies $\sum_{i \in \calS, p} \x_{ijp} \geq 1 - \frac{\delta}{100}$.
		\item Each $T_\ell$ is a $\tilde{O}(1/\delta)$-complete neighborhood with a corresponding set $J_\ell$ of clients responsible for it, and $\Cbb = \cup_{\ell = 1}^L J_\ell$.	
%		
%		
%		\item For the deleted clients $\Cd$, there is a mapping $\phi:\Cd \to \Cbb \cup \Cb$ such that
%		(a) $d(j,\phi(j)) \leq \tilde{O}(1/\delta)$ for all $j\in \Cd$, and
%		(b)	for all $j\in \Cbb \cup \Cb$, we have $\sum_{j' \in \Cd: \phi(j') = j} d_{j'} \leq (1+\delta)\cdot d_j$, i.e., the total demand mapped to $j$ is small.
%		
	\end{enumerate}
\end{theorem}

\begin{proof}[{\bf Proof of Theorem~\ref{fthm:reduction}}]

Let us first describe an approach which fails. Let $(x,y)$ be a feasible solution to LP\eqref{feq:lp1}-\eqref{feq:lp6}, and apply~\Cref{fthm:decomp}.
Although the sets in $\calS$ by definition are roundable which takes care of the clients in $\Cb$, the issue arises in assigning clients of $\Cbb$.
In particular, $\yy_p := \sum_{i\in \calT} y_{ip}$ for all $1\le p\le P$ which indicates the ``supply" of capacity $c_p$ available for the $\Cbb$ clients.
%Since the LP doesn't know $\calT$ beforehand, this supply mayn't be enough.
However, this may not be enough for serving all these clients (even with violation).
That is, the vector $\yy$ may not lie in the (approximate) supply polyhedra of the \cckp instance
defined by $\calT$ as described in Remark~\ref{frem:red}.

That we fail is not surprising; after all, the LP has a bad integrality gap (Remark~\ref{frem:ig}) and we need to strengthen it.
We strengthen the LP by {\em explicitly requiring $\yy$ to be in the supply polyhedra}. Since we do not know $\calT$ before solving the LP (after all the LP generated it),
we go ahead and require this for {\em all} collection of complete-neighborood sets.
More precisely, for $\calT := (T_1,\ldots,T_L)$ of $L$ disjoint complete neighborhood sets, let $\calI_\calT$ denote the \cckp instance a la Remark~\ref{frem:red}.
\begin{equation}\label{feq:lp7}
\forall \calT := (T_1,\ldots,T_L) \textrm{ disjoint neighborhood subsets}, \quad \yy \in \calP(\calI_\calT) \tag{\small{L7}}
\end{equation}
Note that this is a feasible constraint to add to LP\eqref{feq:lp1}-\eqref{feq:lp6}. In the $\opt$ solution, for any $\calT$ there must be enough supply dedicated for the clients responsible for these complete neighborhood sets. So we have the following claim.
\begin{claim}\label{fclm:trivial}
	If $\opt$ is feasible, then there is a feasible solution to LP\eqref{feq:lp1}-\eqref{feq:lp6} along with \eqref{feq:lp7}.
\end{claim}

\noindent
We don't know how (and don't expect) to check feasibility of  \eqref{feq:lp7} for all collections $\calT$. However, we can still run ellipsoid method using the ``round-and-cut'' framework of \cite{CarrFLP00,ChakrabartyCKK11,Li15,Li16}.
To begin with, we start with the LP\eqref{feq:lp1}-\eqref{feq:lp6} and obtain feasible solution $(x,y)$. Subsequently, we apply the decomposition Theorem~\ref{fthm:decomp} to obtain the collection $\calT = (T_1,\ldots,T_L)$.
We then check if $\yy \in \calP(\calI_\calT)$ or not. Since we have a $\gamma$-approximate separation oracle for $\calP(\calI_\calT)$, we  are either guaranteed that $\yy \in \calP(\calI'_\calT)$ where the $\ell^{th}$ demand is now  $D_\ell/\gamma$; or we get a hyperplane separating $\yy$ from $\calP(\calI_\calT)$ which also gives us a
hyperplane separating $y$ from  LP\eqref{feq:lp1}-\eqref{feq:lp7}. This can be fed to the ellipsoid algorithm to obtain a new iterate $(x,y)$ and the above process is repeated. The analysis of the ellipsoid algorithm %\comment{deepc: need to be careful and correct here}
 tells us that in polynomial time we either prove infeasibility of the system \eqref{feq:lp1}-\eqref{feq:lp7} (implying the $\opt$  guess for \mckc is infeasible), or we
are have $(x,y)$ satisfying the premise of the following lemma.
%obtain a solution $(x,y)$
%which is feasible for LP\eqref{feq:lp1}-\eqref{feq:lp6}, such that for the $\calS, \calT$ obtained via Decomposition Theorem~\ref{fthm:decomp}
%and the instance $\calI_\calT$ so obtained, the solution  $\yy \in \calP(\calI'_\calT)$ for the $\gamma$-shaded instance.

\begin{lemma}
	Given $(x,y)$ feasible for  LP\eqref{feq:lp1}-\eqref{feq:lp6}, let us apply the Decomposition Theorem~\ref{fthm:decomp} to obtain the instance $\calS,\calT$.
	Suppose the solution $\yy_p := \sum_{i\in\calT} y_{ip}$ lies in $\calP(\calI'_\calT)$ for the \cckp (respectively, $Q||C_{min}$) instance $\calI'_\calT$ with $L$ machines with $D_\ell := |J_\ell|/\gamma$
	and $f_\ell := |T_\ell|$ (respectively, no cardinality constraints). Then we can obtain an $(\tilde{O}(1/\delta), \beta\gamma(1+\delta))$-approximate solution to the \mckc problem (respectively, with soft-capacities).
\end{lemma}
\begin{proof}
%Suppose the latter occurs. From this we can obtain the bicriteria solution to the \mckc problem as follows.
%Recall the notation in the Decomposition Theorem~\ref{fthm:decomp}.
Since every set $S_k, 1\leq k\leq K,$ is $(\tilde{O}(1/\delta),(1+\delta))$-roundable, there exists a rounding $\y_{ip}$ for $i\in S_k$  such that
\begin{equation}\label{feq:repeat}
\textstyle \forall p, ~~~ \sum_{q\geq p} \sum_{i\in S_k} \y_{iq} \leq \floor{\sum_{q\geq p} \sum_{i\in S_k} y_{iq}}
\end{equation}
\def\s{\tilde{s}}
Ideally, we would like to open a facility of capacity $c_p$ at location $i$ whenever $\y_{ip} = 1$. Unfortunately, the decomposition theorem doesn't have capacity constraints for individual $p$'s but only their suffix sums. Instead we do the following. Define $y^{\calS}_p := \sum_{i\in \calS} y_{ip}$; LP\eqref{feq:lp3} implies that for all $p$, $y^\calS_p + y^\calT_p \leq k_p$.
For $1\leq p\leq P$, define $s_p := \sum_{i\in \calS} \y_{ip}$; \eqref{feq:repeat} implies for all $p$, $\sum_{q\geq p} s_q \leq \floor{\sum_{q\geq p} y^\calS_q}$ (since $\floor{a} + \floor{b} \leq \floor{a+b}$.)
\begin{claim}\label{fclm:massmovement}
	Given $(s_1,\ldots,s_P)$ satisfying $\sum_{q\geq p} s_q \leq \floor{\sum_{q\geq p} y^\calS_q}$, there exists $(\s_1,\ldots,\s_P)$ satisfying
	for all $p$, (a) $\sum_{q\geq p} s_q \leq \sum_{q\geq p} \s_q \leq \sum_{q\leq p} y^\calS_q$, and (b) $\s_p \leq k_p$.
\end{claim}
\begin{proof}
Simply define $\s_p := \floor{\sum_{q\geq p} y^\calS_q} - \floor{\sum_{q> p} y^\calS_q}$. Therefore, $\sum_{q\geq p} \s_q = \floor{\sum_{q\geq p} y^\calS_q}$ implying (a).
To see (b), note $\s_p \leq \ceil{y^\calS_p} \leq k_p$, where we use the fact $\floor{a+b} \leq \floor{a} + \ceil{b}$ for any non-negative $a,b$.
\end{proof}

The first inequality in (a) implies that at every location with $\y_{ip} = 1$, we can open a facility of capacity $c_q \geq c_p$. This, along with condition (b) of roundable sets (Definition~\ref{fdef:rnding-mkc}),
implies we can find a fractional solution $x_{ijp}$ for $j\in \Cb$ and $(i,p)$ with $\y_{ip} = 1$
such that (a) $\sum_{i\in \calS, p\in [P]} x_{ijp} \geq 1$, (b) $x_{ijp} > 0$ only if $d(i,j) \leq \diam(S_k) \leq \tilde{O}(1/\delta)$, and (c) the capacity violation is $\leq (1+\delta)(1 - \delta/100)^{-1} \leq (1+2\delta)$. Note the second term arises since
from the decomposition theorem we have $\sum_{i\in \calS, p\in [P]} \x_{ijp} \geq 1-\delta/100$.
Thus we have fractionally assigned all $\Cb$ clients to open facilities in $\calS$.
%For all $j\in \Cb$, we assign it fractionally using $\frac{\x}{1-\delta/100}$ to

Define, for $p\in [P]$,  $t_p := k_p - \s_p$, the number of facilities of capacity $c_p$ we can open in $\calT$. Note, by Claim~\ref{fclm:massmovement}, $t_p$'s are non-negative.
\begin{claim}
	$(t_1,\ldots,t_P) \in \calP(\calI'_\calT)$
\end{claim}
%	We now use the upward-feasibility property of $\calP$ to prove that the integral vector $(k^{(T)}_1,\ldots, k^{(T)}_P)$ also lies in $\calP(\calI'_\calT)$.
\begin{proof}
	By the Lemma premise, we have $\yy \in \calP(\calI'_\calT)$. %Using the upward-feasibility property of $\calP$, we get the following.
	%Recall that $\calI'_\calT$ has $L$ machines with demands $D_\ell := |J_\ell|/\beta$ and $f_\ell := |T_\ell|$.
	Now note that for all $p$,
	\[
\textstyle 	\sum_{q\geq p} t_q = \sum_{q\ge p} (k_q - \s_q) \ge \sum_{q\ge p} k_q - \sum_{q\geq p} y^\calS_q \geq  \sum_{q\geq p} y^\calT_q
	\]
	Since $\calP(\calI'_\calT)$ is upward-feasible, and $\yy\in \calP$, we get the claim.
\end{proof}
Since $\calP(\calI'_\calT)$ is $\beta$-approximate, we can find an allocation of the $t_p$ copies of jobs of capacity $c_p$ to the $L$ machines of $\calI'_\calT$ such that machine $\ell$ gets at most $f_\ell$ jobs and total capacity $\geq D_\ell/\beta = |J_\ell|/\beta\gamma$. We install these capacities on the facilities of $T_\ell$. Since the diameter of each $T_\ell$ is $\tilde{O}(1/\delta)$, we can find an $x_{ijp}$ assignment of $\Cbb$-clients  to these such that $\sum_{i\in \calT,p\in [P]} x_{ijp} \geq 1$ and $x_{ijp} > 0$ iff $d(i,j) = \tilde{O}(1/\delta)$, such that the capacity violation is at most $\alpha\beta$. This takes care of the clients in $\Cbb$. Finally, Claim~\ref{fclm:prelim3} takes care of all the deleted clients $\Cd$ with an extra hit of $(1+\delta)$ on the capacity and additive $\tilde{O}(1/\delta)$ on the distance.
\end{proof}
\noindent
This completes the proof of Theorem~\ref{fthm:reduction} for the general \mckc problem. For the problem with soft capacities, the proof is exactly the same, except in the end, the instance $\calI_\calT$ is a $Q||C_{min}$ instance rather than a \cckp one.
\end{proof}

We end this section by noting that for the \mckc problem with soft-capacities, if we use the assignment supply polyhedra described in Section~\ref{fsec:supplypolyhedra}, then we do not need to run the ellipsoid algorithm.
In particular, the inequality \eqref{feq:lp7} is implied \eqref{feq:lp1}-\eqref{feq:lp6} for $\calP_\mathsf{ass}$ defined in \eqref{feq:asslp1}-\eqref{feq:asslp3}.
\begin{lemma}\label{flem:implied}
Given any $(x,y)$ feasible for LP\eqref{feq:lp1}-\eqref{feq:lp6} and any $\calT = (T_1,\ldots,T_m)$, we have $\yy \in \calP_\mathsf{ass}(\calI_\calT)$.
\end{lemma}
\begin{proof}
	Fix $\calT = (T_1,\ldots,T_m)$ to be a collection of complete neighborhood sets. In the instance $\calI_\calT$ of $Q||C_{min}$, we have $m$ machines with demands
	$D_\ell = |J_\ell|$, where $J_\ell$ is the client set responsible for $T_\ell$. Recall, $y^\calT_p := \sum_{i \in \calT} y_{ip}$, and we need to find $z_{\ell,p}$ which satisfy the constraints
	\eqref{feq:asslp1}-\eqref{feq:asslp3} where $s_p := y^\calT_p$.
	
	The definition is natural: $z_{\ell,p} := \sum_{i\in T_\ell} y_{ip}$. Clearly it satisfies \eqref{feq:asslp1} (indeed with equality). We now show it satisfies \eqref{feq:asslp2}.
	To this end, define for any $j\in J_\ell$, $x_{jp} := \sum_{i\in T_\ell} x_{ijp}$. Since $\Gamma(J_\ell) \subseteq T_\ell$, we get from \eqref{feq:lp1} that $\sum_p x_{jp} \geq 1$.
	In particular,
	\begin{equation}
\textstyle 	\label{feq:hohum1}\sum_p \sum_{j\in J_\ell}x_{jp} \geq D_\ell
	\end{equation}
	
	From \eqref{feq:lp4}, we know $x_{ijp} \leq y_{ip}$ and summing over all $i\in T_\ell$, we get for all $j\in T_\ell$, $x_{jp} \leq \sum_{i\in T_\ell} y_{ip} = z_{\ell,p}$.
	In particular, $\sum_{j\in J_\ell} x_{jp} \leq z_{\ell,p} D_\ell$. From \eqref{feq:lp2} we know for all $i\in T_\ell, p\in [P]$, $\sum_{j\in J_\ell} x_{ijp} \leq c_p y_{ip}$.
	Summing over all $i\in T_\ell$, gives $\sum_{j\in J_\ell} x_{jp} \leq c_p z_{\ell,p}$. Putting together, we get
	\begin{equation}
	\textstyle 	\label{feq:hohum2}\sum_{j\in J_\ell}x_{jp} \leq z_{\ell,p} \min(D_\ell,c_p)
	\end{equation}
	\eqref{feq:hohum1} and \eqref{feq:hohum2} imply that $z$ satisfies \eqref{feq:asslp2}.
\end{proof}
Therefore, one can use the natural LP relaxation to obtain for any $\delta > 0$, a $\left(\tilde{O}(1/\delta),(2+\delta)\right)$-bicriteria approximation for the \mckc problem with soft capacities.
As it should be clear, this is a much more efficient algorithm.

%\newpage
\section{Proof of Decomposition Theorem~\ref{fthm:decomp}} \label{fsec:decomp-proof}
\def\F{F_\mathsf{tentative}}
\noindent
%{\bf Proof of Theorem~\ref{fthm:decomp}}\
%\comment{\bf \large Need to provide a sketch-paragraph. }\\
We now prove the Decomposition Theorem~\ref{fthm:decomp}. We first describe the algorithm which constructs the partitions into roundable and complete neighborhood sets. It is based on the following refinement of the region growing idea used in the proof of Theorem~\ref{fthm:weakdecomp} -- starting
with an arbitrary client we first check if there is a small enough neighborhood (i.e., of small diameter)  around it which is {\em non-expanding}, i.e., the number of clients on the boundary are much smaller than the number of clients inside the neighborhood. If so, we can remove the clients on the boundary and obtain a complete neighborhood set. Otherwise, we show that the total $y$-mass of the facilities in this neighborhood is quite high, and so, we can get
a roundable set of facilities.   The algorithm is written formally in Algorithm~1. We  analyze the algorithm subsequently  and show that it has the desired properties.
Throughout, we let $\epsilon := \min(1/12,\delta/100)$.
%We prove this theorem by describing an algorithm which constructs these partitions.
%The algorithm is written formally in Algorithm~1. We describe it in detail first.

\subsection{Algorithm Description}
Our algorithm starts with the collections $\calS$ and $\calT$, and the clients sets $\Cd$, $\Cbb$, and $\Cb$ being empty. Once a facility is assigned into a set in $\calS$ or $\calT$, it is called an \emph{assigned facility}. Similarly clients are assigned once they are added to $\Cd \cup \Cbb \cup \Cb$.
As our algorithm forms these clusters, it changes the connection graph $G$ by deleting all assigned clients and facilities. At any time, we denote the residual graph by $H$.
We make a couple of definitions as in \Cref{fsec:regiongrowing}. Given a node $i$ and an integer $t$, let $\Gamma^{(t)}_H(i)$ denote the nodes at distance (in $H$) exactly $t$ from $i$. We let $N^{(t)}_H(i)$ denote the nodes at distance $< t$ from $i$.
We use the shorthand $\Gamma_H(i)$ to denote $\Gamma^{(1)}_H(i)$. We extend this definition to subsets: $\Gamma_H(S) := \cup_{i\in S}\Gamma(i)$.
%The neighborhood structure in the residual graph is denoted by $\Gamma_H(\cdot)$, e.g., $\Gamma_H(i)$ denotes the set of neighbors of $i$ in $H$, and $\Gamma_H(S)$ denotes the neighbors of a set of vertices $S$ in $H$.
Since we only delete vertices from the graph over the iterations, $\Gamma_H(S) \subseteq \Gamma_G(S)$ for all sets $S \subseteq V$ of the original set of vertices of $G$.
For each of the partitions $\calS$ and $\calT$, let $L(\calS) = \cup_{1\leq k \leq K} S_k$ and $L(\calT) = \cup_{1 \leq \ell \leq L} T_\ell$ denote the set of all locations in them respectively.
Each set $S_k$ (resp. $T_\ell$) in the partitions $\calS$ (resp. $\calT$) will have a {\em root} facility $i_k \in S_k$ (resp. $i_\ell \in T_\ell$). We use $R(\calS)$ and $R(\calT)$ to denote the collection of roots $\cup_{1 \leq k \leq K} \{i_k\}$ and $\cup_{1 \leq \ell \leq L} \{i_\ell\}$ in $\calS$ and $\calT$ respectively. \smallskip

A key definition in our algorithm is that of {\em effective capacity}. For every $i\notin L(\calS)\cup L(\calT)$ and $p\in [P]$ with $y_{ip} > 0$, define
\[
\effc(i,p) := \frac{\sum_{j\in H \cap C} d_j x_{ijp}}{y_{ip}}
\]
Recall that $C$ denotes the set of all clients, and therefore, $H \cap C$ is the set of unassigned clients.
Since all sets
 are initially empty, $\effc(i,p)$ is well defined for all $i\in F, p\in [P]$. Whenever a facility enters $L(\calS)\cup L(\calT)$, we fix its $\effc(i,p)$ to be what it was at the iteration it entered.
Since the set of clients in $H$ only monotonically decreases, the effective capacity can only decrease over time. Each iteration of the algorithm (\Cref{falg:iter}) begins by picking the pair $(i^\star, p^\star)$ with the highest
effective capacity (\Cref{falg:istar}). %We let $J_1$ denote the neighbors of $i^\star$, that is, $J_1 := \Gamma_H(i^\star)$; note that by definition of $\effc(i^\star,p^\star)$ we get $|J_1| > \effc(i^\star,p^\star)$ (see Claim~\ref{fclm:imp}).

\def\t{\bar{t}}
Let $t^\star$ be the smallest {\em even} integer $> \ceil{\frac{8}{\eps}\ln\left(\frac{1}{\eps}\right)}$. We set $\t$ to be the smallest odd number $t$ in $\{1,\ldots,t^\star\}$ such that $|\Gamma^{(t)}_H(i^\star)\cap C| < \eps\cdot |N^{(t)}_H(i^\star)\cap C|$ if such a number exists, otherwise we set $\t = t^\star + 1$.  That is, as in Section~\ref{fsec:regiongrowing} we find the smallest distance at which the ``client ball'' stops expanding, however, we stop once we cross $t^\star$. Depending on what $\t$ is (although note that it is alway an odd number), we have two cases.
%
%Starting from $i^\star$, we look at nodes in $H$ at distance at most 4 from it. Since
%$H$ is a bipartite graph, we alternate between location vertices and client vertices. We denote $J_1$ and $J_2$ as clients which are at distance 1 and 3 from $i^\star$ respectively. Similarly, let $A$ and $B$ be locations at distance 2 and 4 from $i^\star$ respectively(\Cref{falg:j1}-\Cref{falg:B}). Note that clients in $J_1$ are adjacent to $i^\star$ and locations in $A$ in $H$, and those in $J_2$ are adjacent to locations in $A$ and $B$ only. Now, let $D(J_1)$ and $D(J_2)$ denote the
%total demand of $J_1$ and $J_2$ respectively. The algorithms branches in several cases:
\begin{itemize}
\item (If $\t =  t^\star + 1$):  In this case, we have always witnessed expansion. We form a new component  $S_k := N^{(t^\star)}_H(i^\star)\cap F$ to be all the facilities in the $t^\star$-ball around $i^\star$.
Let $J_\mathsf{int} := N^{(t^\star-1)}_H(i^\star)$ be the clients whose neighbors in $H$ lie in $S_k$. Since we witness expansion at all stages, note that
\begin{eqnarray}
\label{feq:jint}
 |J_\mathsf{int}| \geq \left(1+\eps\right)^{t^\star/2}\cdot |\Gamma_H(i^\star)| > \frac{1}{\eps^4}\cdot |\Gamma_H(i^\star)|
\end{eqnarray}
It is not too hard to see that $|\Gamma_H(i^\star)| \ge \effc(i^\star,p^\star)$ (see Claim~\ref{fclm:imp}).
Therefore, the  (fractionally opened) facilities in $S_k$ are servicing a large enough demand, in particular, more than the effective capacity of $i^\star$ and by the greedy choice, the effective capacity of any $(i,p)$.
This implies there will be considerable mass ($>1/\eps$)  of facilities of the same type in $S_k$ opened fractionally; opening floor-many of them violates capacity by only $(1+\eps)$.
So, we add $S_k$ to $\calS$, make $i^\star$ its root and add $i^\star$ to $R(\calS)$  (\Cref{falg:case1}). We remove $S_k \cup J_\mathsf{int}$ from $H$ and add $J_\mathsf{int}$ to $\Cb$. Additionally, for a technical reason, we remove from $H$ any other client $j$
with $\sum_{i\in \calS, p} x_{ijp} > (1-\eps)$; in this case we set $\x_{ijp} = x_{ijp}$ for all $i\in \calS, p\in [P]$ and $\x_{ijp} = 0$ for all $i\notin \calS, p\in [P]$ and add $j$ to $\Cb$~(\Cref{falg:blue}).
% By the greedy choice, the effective capacities of all these facilities is at most that of $i^\star$. Thus,
%this corresponds to a situation where we have large amount of capacity generated by a set of facilities which are fractionally open, but the {\em integral} capacity of each of these is small. In this case, we will show that one can replace these fractionally open facilities by integral facilities.
%
%
%\item $D(J_2) \geq 32 D(J_1)$: In this case, the (fractionally opened) facilities in $\{i^\star\} \cup A \cup B$ are servicing a large enough demand, in particular, more than the effective capacity of $i^\star$. This is so because
%    the effective capacity of $i^\star$ is about $D(J_1)$, whereas the total demand serviced by $i^\star \cup A \cup B$ is about $D(J_2)$. Now, notice tha Therefore, our algorithm adds the set
%    $ S_k := \{i^\star\} \cup A \cup B$ to $\calS$. It makes $i^\star$ as the root of  $S_k$. The set
%    $S_k$ is now removed from $H$. We will also remove $A$ and $B$ from $H$. In fact, we remove any client which is assigned to a total fraction of more than half to the facilities in $\calS$, and add it to the set~$\Cb$
\item (If $\t \leq t^\star$): Let $J_\mathsf{ext} := \Gamma_H^{(\t)}(i^\star)$ and let $J_\mathsf{int} := N_H^{(\t)}(i^\star) \cap C$. We know that $|J_\mathsf{ext}| < \eps\cdot |J_\mathsf{int}|$.
Let $\F := N_H^{(\t)}(i^\star) \cap F$ be the facilities in this ball. We delete $J_\mathsf{ext}$ from $H$ and add it to $\Cd$; we can do so since we can ``charge it'' to $J_\mathsf{int}$.
Ideally, we would like to add $(\F,J_\mathsf{int})$ as a complete-neighborhood to $\calT$.
While it is true that $\Gamma_H(J_\mathsf{int}) \subseteq \F$, the same may not be true in the original graph $G$ since we delete vertices from it. More precisely, there could be a client $j\in J_\mathsf{int}$ and a facility $i\in \calS\cup\calT$
such that $(i,j) \in G$. Therefore, the algorithm branches into two sub-cases.
%
%
%In this case, we would like to delete $J_2$ and assign them to $J_1$. Further, we can would like to treat $\{i^\star\} \cup A \cup B$ as a 4-complete neighborhood of $J_1$ (and then add these facilities as a new set in $\calT$). However, the worry is that the neighborhoods $A,B, J_1, J_2$ are defined with respect to the graph
%    $H$, and not the graph $G$. So it is possible that we have deleted some facility, and a clients in $J_1$ are fractionally assigned to such facilities. Therefore, the algorithm considers two sub-cases:

(i) There is some root center  $i_r \in R(\calS)$ close to $i^\star$ (\Cref{falg:case2}). In  this
    sub-case, the algorithm considers the closest such root $i_r$, and {\em augments} $S_r$ to $S_r \cup \F$. As in the above case,
    we update $\Cb$ by adding to it any client which has more than $(1-\eps)$ of its fractional assignment to facilities in $\calS$ (in particular, $J_\mathsf{int}$ will
    get added to this set)

 (ii) There is no such root (\Cref{falg:case3}).  In this case, the set $\F$ gets added as a new set $T_\ell$ to $\calT$. Further,
    we add $J_\mathsf{int}$ to $\Cbb$. One of the invariants of our algorithm is that in later stages when we again encounter this case ($\t \leq t^\star$), any client $j\in J_\mathsf{int}$ at that stage {\em cannot} be a neighbor in $G$ to a facility $i\in \calT$.
    \end{itemize}

This completes the description of the algorithm. We now show that the decomposition has the desired properties; in particular it satisfies the conditions in Theorem~\ref{fthm:decomp}.
 %{\bf Maybe want to define this for $i\notin S$ and fix it for those going in...} \bigskip

\begin{algorithm}
\caption{Rounding algorithm for~\Cref{fthm:decomp}}\label{falground}
\begin{algorithmic}[1]
\Procedure{AlgDecompose}{$x,y$}
\State $t \gets 1$; $k \gets 1$; $\ell \gets 1$; $H  \gets G$; $t^\star \gets$ smallest {\em even} integer $> \ceil{\frac{8}{\eps}\ln\left(\frac{1}{\eps}\right)}$; $\x \gets x$
\While{there are no unassigned clients, i.e., $V(H) \cap C \neq \emptyset$ } \label{falg:iter}
\State {\bf update} $\effc(\cdot,\cdot)$ for all $i \in V(H) \cap F, p \in [P]$
\State $(i^\star,p^\star) \gets \argmax_{i \in H, p \in [P]} \effc(i,p)$ \Comment{\emph{pick location and type with largest effective capacity}} \label{falg:istar}
\State $\t \gets 1$ \Comment{\emph{Find the smallest odd $\t$ which is non-expanding}}
\While{$\t < t^\star$}
\If{$|\Gamma^{(t)}_H(i^\star)| < \eps\cdot |N^{(t)}_H(i^\star)\cap C|$}
\State Exit While Loop
\Else
\State $\t \gets \t+2$
\EndIf
\EndWhile
%\State $J_1 \gets \Gamma_H(i^\star)$ \Comment{\emph{$J_1$ is the set of unassigned client neighbors}} \label{falg:j1}
%\State $A \gets \Gamma_H(J_1) \setminus \{i^\star\}$ \Comment{\emph{$A \cup \{i^\star\}$ is the set of facility neighbors of $J_1$ in $H$}}
%\State $J_2 \gets \Gamma_H(A) \setminus J_1$ \Comment{\emph{$J_1 \cup J_2$ is the set of unassigned client neighbors of $A$}}
%\State $B \gets \Gamma_H(A) \setminus \{A \cup i^\star\}$ \Comment{\emph{$A \cup B \cup \{i^\star\}$ is the set of facility neighbors of $J_1 \cup J_2$ in $H$}} \label{falg:B}
%\State $D(J_1) \gets \sum_{j \in J_1} d(j)$
%\State $D(J_2) \gets \sum_{j \in J_2} d(j)$
\If{$\t = t^\star + 1$} \Comment{\emph{$N^{(t^\star)}_H(i^\star)\cap F$ will be locally roundable}}
%\If{$D(J_2) \geq 32 D(J_1)$} \Comment{\emph{$i^\star \cup A \cup B$ will be locally roundable}}
\State $S_k \gets N^{(t^\star)}_H(i^\star)\cap F$ and $\calS \gets \calS \cup S_k$ \label{falg:case1}
\State {\bf define} $i^\star$ to be the root of $S_k$, i.e., $R(\calS) \gets R(\calS) \cup \{i^\star\}$
\State $H \gets H \setminus S_k$ \Comment{\emph{remove assigned facilities from $H$}}
\State $J_\mathsf{int}\gets  N^{(t^\star-1)}_H(i^\star)$ \Comment{\emph{Note $|J_\mathsf{int}| \geq \frac{1}{\eps^4}\cdot |\Gamma(i^\star)| $}}
\For{each $j \in H$ s.t $\sum_{i \in \calS, p} x_{ijp} > (1-\eps)$} \Comment{\emph{In particular, this contains $J_\mathsf{int}$}}
\State $\Cb \gets \Cb \cup \{j\}$ and $H \gets H \setminus \{j\}$ \Comment{\emph{assign clients to $\Cb$}} \label{falg:blue}
\State $\x_{ijp} \gets 0$ for all $i\notin \calS, p\in [P]$ \Comment{\emph{Set $\x$ to $0$ for facilities not in $\calS$}} \label{falg:setx}
\EndFor
\State $k \gets k + 1$
\Else \Comment{\emph{$\t < t^\star$, i.e., the ball is non-expanding.}}
\State $\F \gets N_H^{(\t)}(i^\star) \cap F$ \Comment{\emph{$\F$ are the ball's facilities.}}
\State $J_\mathsf{ext} \gets \Gamma_H^{(\t)}(i^\star)$ \Comment{\emph{Ball's boundary clients}}
\State $J_\mathsf{int} \gets N_H^{(\t)}(i^\star) \cap C$. \Comment{\emph{Ball's internal clients}}
\State $\Cd \gets \Cd \cup J_\mathsf{ext}$ and define $\phi$ appropriately \Comment{\emph{delete $J_\mathsf{ext}$ and charge to $J_\mathsf{int}$}} \label{falg:phi1}\label{falg:phi2}
\State $H \gets H \setminus J_\mathsf{ext} $ \Comment{\emph{remove deleted clients from $H$}}
\If{$\dist_G(i^\star, R(\calS)) \leq \frac{16}{\eps}\ln\left(\frac{1}{\eps}\right)$} \Comment{\emph{$i^\star$ is close to some root in $R(\calS)$}} \label{falg:case2}

\State {\bf let} $i_r = \argmin_{i \in R(\calS)} \dist_G(i^\star,i)$ \Comment{\emph{$i_r$ is the nearby root from $\calS$}}
\State $S_r \gets S_r \cup \F$ \Comment{\emph{add these facilities to $S_r$}} \label{falg:case2a}
\For{each $j \in H$ s.t $\sum_{i \in \calS, p} x_{ijp} > (1-\eps)$} \Comment{\emph{In particular, this contains $J_\mathsf{int}$}}
\State $\Cb \gets \Cb \cup \{j\}$ and $H \gets H \setminus \{j\}$ \Comment{\emph{assign clients to $\Cb$}} \label{falg:case2ac}
\State $\x_{ijp} \gets 0$ for all $i\notin \calS, p\in [P]$ \Comment{\emph{Set $\x$ to $0$ for facilities not in $\calS$}} \label{falg:setx-2}
\EndFor
\Else \Comment{\emph{$\F$ will be a $\tilde{O}(1/\eps)$-complete neighborhood of $J_\mathsf{int}$}} \label{falg:case3}
\State {\bf add} a new part $T_\ell := \F$ to $\calT$  \label{falg:case2b}
%\State $H \gets H \setminus (i^\star \cup A)$ \Comment{\emph{delete assigned facilities from $H$}}
\State $J_\ell \gets J_\mathsf{int}$, $\Cbb \gets \Cbb \cup J_\mathsf{int}$, and $H \gets H \setminus J_1$ \Comment{\emph{assign clients to $\Cbb$}} \label{falg:color2}
\State $\ell \gets \ell + 1$
\EndIf
\State $H \gets H \setminus \F$ \Comment{\emph{remove assigned facilities  from $H$}}
\EndIf
\State $t \gets t+1 $ \Comment{\emph{Iteration Counter}}
\EndWhile\label{feuclidendwhile}
\State \textbf{return} $\calS,\calT$
\EndProcedure
\end{algorithmic}
\end{algorithm}
\subsection{Algorithm Analysis}
%In this section, we analyze the algorithm.
Firstly, the statement of the algorithm implies the partition $\calS, \calT$ and $\Cb\cup\Cbb\cup \Cd$.
We analyze the algorithm to prove the properties needed.
\medskip \noindent
At the beginning of each iteration, we want to show that the algorithm maintains the following invariants:
\begin{framed}
\begin{enumerate}%[noitemsep]
	\item[I1.] For any facility $i\in L(\calT)$, $\Gamma_G(i) \cap V(H) = \emptyset$, i.e., $\Gamma_G(i)$  contains no unassigned clients. Note that this holds even w.r.t. all the neighbors according to the original graph $G$.
\item[I2.] Similarly, for any facility $i \in L(\calS)$  added in~\Cref{falg:case2a} in Algorithm~1, $\Gamma_G(i)$ contains no unassigned clients.
\end{enumerate}
\end{framed}
Note that in I2, we count only those $i$ which get added to $L(\calS)$ in~\Cref{falg:case2a}, and so do not consider locations getting added in~\Cref{falg:case1}.

\begin{claim}
\label{fcl:inv}
The two invariants hold at the beginning of every iteration of the while loop in~\Cref{falg:iter}.
\end{claim}
\begin{proof}
We show this by induction over the number of iterations $t$. Clearly, at $t=1$, $L(\calT)$ and $L(\calS)$ are empty, so the invariants hold tautologically. Suppose they hold for iterations upto $i$. We show that they also hold at the end of the $t^{th}$ iteration, and hence they hold at the beginning of the $(t+1)^{th}$ iteration, thus completing the proof. To this end, consider the $t^{th}$ iteration.

\medskip \noindent We first show that I1 continues to hold at the end of this iteration. Note that we only need to check if I1 holds for any new facilities added to $L(\calT)$ in this iteration, which only happens in~\Cref{falg:case2b}. In this case, consider any facility $i \in T_\ell$, the set of facilities added to $L(\calT)$, and consider the neighborhood $\Gamma_G(i)$: in this set, some clients are already in $\Cb \cup \Cbb \cup \Cd$ in which case they would have been deleted from $H$ in earlier iterations. By definition, the remaining clients belong to $J_\mathsf{int} \cup J_\mathsf{ext}$, since $J_\mathsf{int} \cup J_\mathsf{ext}$ contains all remaining neighbors of $T_\ell$. But clients in $J_\mathsf{int} $ are added to $\Cbb$, and those in $ J_\mathsf{ext}$ are added to $\Cd$, hence $i$ would have no  clients as neighbors in $H$ at the end of this iteration. Applying this to all $i \in T_\ell$ completes the proof.

\medskip \noindent We now show that I2 continues to hold at the end of this iteration. Similar to the above proof, note that we only need to check if I2 holds for any new facilities added to $L(\calS)$ in~\Cref{falg:case2a}. In this case, consider any facility $i \in \F$, the set of facilities added to $L(\calS)$, and consider the neighborhood $\Gamma_G(i)$: in this neighborhood, some clients are already in $\Cb \cup \Cbb \cup \Cd$ in which case they would have been deleted from $H$ in earlier iterations. By definition, the remaining clients belong to $J_\mathsf{int} \cup J_\mathsf{ext}$, since $J_\mathsf{int} \cup J_\mathsf{ext}$ contains all remaining neighbors of $\F$. But clients in $J_\mathsf{ext}$ are added to $\Cd$, and we now show that all clients in $J_\mathsf{int}$ would be colored blue in~\Cref{falg:case2ac}, hence showing that $i$ would have no clients as neighbors in $H$ at the end of this iteration.
Indeed, consider any client $j \in J_\mathsf{int}$: by definition, it was in $H$ at the beginning of this iteration and so by invariant I1, there are no edges in $H$ between $j$ and any location $i' \in L(\calT)$. So all neighbors in $\Gamma_G(j)$ which have already been deleted belong to $L(\calS)$. Moreover, $\F$ includes all remaining neighbors of $j$. Hence, for any such $j$, we know that $\sum_{i \in L(\calS), p} x_{ijp} = 1$, and so it would be
added to $\Cb$ in~\Cref{falg:case2ac}.
\end{proof}

We now show that the deleted clients $\Cd$ can be charged to $\Cb$ and $\Cbb$.
\begin{claim} \label{fcl:phi-augment}
$\Cd$ is a $(\tilde{O}(1/\delta),\delta)$-deletable set.
%In~\Cref{falg:phi1}, $\phi$ can be augmented so that the Property~(4) of~\Cref{fthm:decomp} is satisfied.
\end{claim}
\begin{proof}
	We add vertices to $\Cd$ only in~\cref{falg:phi1}, and at that point it must be that $|J_\mathsf{ext}| \leq \eps\cdot |J_\mathsf{int}|$.
	As in the proof of Theorem~\ref{fthm:weakdecomp}, we can define the assignment $\phi_{j,j'}$ for $j\in J_\mathsf{ext}$ and $j' \in \mathsf{int}$.
%	So indeed we are justified in deleting the clients in $J_\mathsf{ext}$ and ``charging'' them to $J_\mathsf{int}$. More precisely, we augment the
%mapping function $\phi:J_\mathsf{ext} \to J_\mathsf{int}$ such that for all $j\in J_\mathsf{int}, \sum_{k:\phi(k) = j} d_k \leq \eps d_j$. Again, this is possible since $|J_\mathsf{ext}| < \eps|J_\mathsf{int}|$. Also note for any $j\in J_\mathsf{ext}$ and $k\in J_\mathsf{int}$,
%we have $d(j,k) \leq \tilde{O}(1/\epsilon)$.
Furthermore, as in the proof of Claim~\ref{fcl:inv}, our algorithm makes sure that the client set $J_\mathsf{int}$ gets added to $\Cb$ or $\Cbb$ (in~\cref{falg:case2ac,falg:color2}).
Therefore, these clients in $J_\mathsf{int}$ will never be images of $\phi$ again, thus completing the proof.
\end{proof}

We will now show that the sets $\{T_\ell\}$ in the family $\calT$ form $\tilde{O}(1/\epsilon)$-complete neighborhoods supported by the corresponding client-sets $\{J_\ell\}$.
\begin{lemma} \label{flem:local}
Consider an iteration when a new set $T_\ell$ is added to $\calT$ in~\cref{falg:case2b}. The set $T_\ell$ is a $\tilde{O}(1/\epsilon)$-complete neighborhood supported by the set of clients $J_\ell$ (defined in~\cref{falg:color2}).
\end{lemma}

\begin{proof}
Firstly, the diameter of the new set is at most $\frac{16}{\epsilon}\ln(1/\epsilon)$, since for every $i \in T_\ell$ is $d(i,i^\star) \leq t^\star$ for the  $i^\star$ facility identified in~\cref{falg:istar}. To complete the proof, we show that $T_\ell$ is supported by the set $J_\ell$ defined in~\cref{falg:color2}, which is same as $J_\mathsf{int}$. % computed in~\cref{falg:j1}.
We establish this by showing that $\Gamma_G(J_\mathsf{int}) \subseteq T_\ell$ (recall~\cref{fdef:comp-nbr}).

To this end, consider a client $j \in J_\mathsf{int}$. At the beginning of this iteration, $j$ is a client  in $H$.
We claim that at the beginning of this iteration, $\Gamma_H(j) = \Gamma_G(j)$ (i.e., no neighboring facility has already been assigned in earlier iterations). Indeed, suppose not, and let $i$ be some facility which is present in $\Gamma_G(j)$ but not in $\Gamma_H(j)$. We first observe that $i$ cannot be in $L(\calT)$ as that would violate invariant I1 at the beginning of this iteration --- $(i,j)$ would form the violated pair. Similarly, we note that $i$ cannot be added to $\calS$ in~\cref{falg:case1} in an earlier iteration --- because then the distance between $i^\star$ and $R(\calS)$ would be at most $\frac{16}{\epsilon}\ln(1/\epsilon)$ (via the path $i^\star \rightarrow j \rightarrow i \rightarrow R(\calS)$), so this is a contradiction to the fact that the algorithm is in the branch executing~\cref{falg:case2b}. Finally, we note that $i$ cannot be added to $\calS$ in~\cref{falg:case2a} in an earlier iteration, as that would violate invariant I2 at the beginning of this iteration --- again $(i,j)$ would form the violated pair.
So we can conclude that $\Gamma_H(j) = \Gamma_G(j)$ and thus that the entire neighborhood of $j$ is contained in $\{i^\star\} \cup A$ which is added to $T_\ell$ in this iteration. Repeating this argument for all $j$ shows that $\Gamma_G(J_\mathsf{int}) \subseteq T_\ell$.
\end{proof}

\medskip \noindent We now turn our attention to proving that the sets in $\calS$ are locally roundable.
 Toward this end, we begin with the following useful claim.

\begin{claim}\label{fclm:imp}
%Consider an iteration when a new set $S_k$ is added to $\calS$ in~\cref{falg:case1}.
For any set $S_k \in \calS$, we have  $\sum_{j\in C} \sum_{i\in S_k}  \sum_{p\in [P]} d_j \x_{ijp} \geq \frac{1}{\epsilon^3}\cdot \max_{i\in S_k, p\in [P]} \effc(i,p) $
\end{claim}
\begin{proof}
Given a set $S_k \in \calS$, there is a subset $S'_k$ which was formed in~\cref{falg:case1} and then got augmented in~\cref{falg:case2a}.
Note that the root $i^\star$ of $S_k$ lies in $S'_k$. First we prove the claim for the set $S'_k$.
By the definition of $(i^\star, p^\star)$~(\cref{falg:istar}), we know that $ \max_{i\in S'_k, p\in [P]} \effc(i,p)  = \effc(i^\star,p^\star)$.% in the iteration when $S_k$ was added to $\calS$.
Furthermore, by definition,
\[\effc(i^\star,p^\star) = \sum_{j\in H\cap C} d_j \frac{\x_{i^\star jp^\star}}{y_{i^\star p^\star}} \leq \sum_{j\in \Gamma_H(i^\star)} d_j = |\Gamma_H(i^\star)|\]
The inequality follows because we know  (a) that $\x_{ijp^\star} = x_{ijp^\star} > 0$ only for $j\in \Gamma_H(i^\star)$ (i.e., only clients which are neighboring $i^\star$ can be serviced by $i^\star$), and (b) that $x_{i^\star jp^\star} \leq y_{i^\star p^\star}$ (using inequality~\eqref{feq:lp4}).
Now note that for any $j\in J_\mathsf{int}$, since $j\in H$, we have that $\sum_{i\in L(\calS), p\in [P]} x_{ijp} \leq 1 - \eps$ (otherwise it would have been added
to $\Cb$ in an earlier iteration and deleted from $H$), and so $\sum_{i\in S_k, p\in [P]} x_{ijp} \geq \eps$ (because $S_k$ includes all the neighbors of $j$ not already in $L(\calS)$, and $j$ has no edge to any facility in $L(\calT$) even in the original graph $G$ by invariant I1, so the fractional demand from $j$ to vertices in $L(\calT)$ is $0$).
Therefore,
\[
|J_\mathsf{int}| = \sum_{j\in J_\mathsf{int}} d_j  \leq \frac{1}{\eps} \sum_{j\in J_\mathsf{int}} d_j \left(\sum_{i\in S_k, p\in [P]} x_{ijp}\right)
\]
Since  $|J_\mathsf{int}| \geq \frac{1}{\eps^4}\cdot |\Gamma(i^\star)| $~\eqref{feq:jint}, we have the claim for $S'_k$.

	Subsequently, since $i^\star$ got added to $H$, $\effc(i^\star,p^\star)$ remains unchanged and since $\effc(i,p)$ can only decrease, we see $\effc(i^\star,p^\star) =  \max_{i\in S_k,p\in [P]} \effc(i,p)$.
	Therefore, even when we add more facilities
	to $S_k$ later in the algorithm~(\cref{falg:case2a}), the RHS of the inequality in the statement of the Claim does not change.
	Observe that the LHS can only increase since the set $S_k$ can grow during the algorithm. \end{proof}

\begin{claim}
		At the end of the algorithm, for all $j \in \Cb$, $\sum_{i \in \calS, p \in [P]} \x_{ijp} > (1-\frac{\delta}{100})$.	
\end{claim}
\begin{proof}
This follows since we only add clients to $\Cb$ when their fractional allocation to $\calS$ exceeds $(1-\epsilon)$.
\end{proof}

\begin{lemma}
\label{flem:sk}
Each set $S_k \in \calS$ is a $\left(\tilde{O}\left(\frac{1}{\delta}\right),(1+\delta)\right)$-roundable set with respect to $(\x,y)$.
\end{lemma}
\begin{proof}
(Diameter) We claim that $\diam(S_k) \leq \frac{50}{\epsilon}\ln(\frac{1}{\epsilon})$ for every $S_k \in \calS$. We show by induction that for each $S_k \in \calS,$
$\dist_G(i,i_k) \leq \frac{25}{\epsilon}\ln(\frac{1}{\epsilon})$ for every $i \in S_k$, where $i_k$ is the root of $S_k$. When we add a new set $S$ to $\calS$ (in~\Cref{falg:case1}),
$S$ is the set $N^{(t^\star)}_H(i^\star)\cap F$ and so $\dist_G(i,i^\star) \leq t^\star \leq \frac{9}{\epsilon}\ln(\frac{1}{\epsilon})$ % Clearly, $\dist_G(i, i^\star) \leq 4$
for any $i \in S$. Now, consider the case when we augment an existing set  in $\cal S$ (as in~\Cref{falg:case2a}).
Again, using the notation in the algorithm, suppose $i_k = \argmin_{i \in R(\calS)} \dist_G(i, i^\star)$, and let $i_k$ be the
root of $S_k \in \calS$. Then, $\dist_G(i^\star, i_k) \leq \frac{16}{\epsilon}\ln(\frac{1}{\epsilon})$. Since $\dist_G(i^\star, i') \leq \frac{9}{\epsilon}\ln(\frac{1}{\epsilon})$ for any $i' \in \F$, we see that
$\dist_G(i_k, i') \leq \frac{25}{\epsilon}\ln(\frac{1}{\epsilon})$ for any $i' \in \F$.
So the desired claim follows by induction. Since $\eps = O(\delta)$, the diameter condition follows.
\smallskip

\noindent
(Roundability)
We now show that there is a rounding of $\y_{ip}$ values for $i\in S_k$ such that
		\begin{enumerate}
			\item $\sum_{q \geq p} \sum_{i\in S_k} \y_{iq} ~\leq~ \floor{\sum_{q \geq p}\sum_{i\in S} y_{iq}}$, and
			\item $\sum_{j\in C} d_j \sum_{i\in S_k,p\in [P]} \x_{ijp} \leq (1+\delta)\cdot \sum_{i\in S} \sum_{p\in [t]} c_p\y_{ip}$
		\end{enumerate}
%To do so, we first claim
%\begin{claim}\label{fclm:imp2}
%	%For any $S\in \calS$, we have
%	 $\sum_{j\in C}\sum_{i\in S_k}\sum_{p\in [P]} d_jx_{ijp} \geq 16\max_{i\in S_k,p\in [P]} \effc(i,p)$.
%\end{claim}
%\begin{proof}
%	When the set $S_k$ is first formed (\Cref{falg:case1}), this follows from Claim~\ref{fclm:imp}. Let $(i_k, p_k)$ be the pair
%	with the highest effective capacity which was selected in the iteration when $S_k$ was formed. Then $i_k$ is the
%	root of $p_k$, and $\effc(i_k, p_k) = \max_{i\in S_k,p\in [P]} \effc(i,p)$. Henceforth, $\effc(i_k, p_k)$ does not
%	change. Further, $\effc(i_k, p_k) \geq \effc(i,p)$ for any $i \in H, p \in [P]$. Since $\effc(i,p)$ can never increase as the
%	algorithm progresses, this inequality will hold from this iteration onwards. Therefore, even when we add more facilities
%	to $S_k$ later in the algorithm~(\cref{falg:case2a}), the RHS of the inequality in the statement of the Claim does not change.
%	Observe that the LHS can only increase since the set $S_k$ can grow during the algorithm.
%\end{proof}

\noindent
%Now we are armed to prove the roundability property.
For simplicity, let us use $\Delta := (1+\epsilon)$.
Define $A_u := \{(i,p), i\in S_k,p\in [P]: \effc(i,p) \in [\Delta^u,\Delta^{u+1})\}$ and let $\max_{i\in S_k,p\in [P]} \effc(i,p) \in [\Delta^U,\Delta^{U+1})$.
From Claim~\ref{fclm:imp}, we have
\[
\textstyle \Delta^U \leq ~ \max_{i\in S_k,p\in [P]} \effc(i,p) \leq ~\epsilon^3 \sum_{j\in C} d_j \sum_{i\in S_k,p\in [P]} \x_{ijp}
\]
We also assume we have available capacities of value $\Delta^u$ available to us; this is without loss of generality by setting there $k_p$ value to $0$
if there don't exist any.

Define $\alpha_u := \sum_{(i,p) \in A_u} y_{i,p}$. For all values of $u$,
%Clearly, $|A_u| \geq \alpha_u$. Therefore, we can choose a subset $F_u$ of size $\floor{\alpha_u/2}$ from $A_u$.
{\em arbitrarily} choose $\floor{\alpha_u}$ different facilities $F_u$ in $S_k$; that there are so many is implied by \eqref{feq:lp4} of the LP.
For each $u$ and for each $i\in F_u$, set $\y_{i\Delta^u} = 1$. For every other $(i,p)$, set $\y_{ip} = 0$.
Note that  $\sum_{i\in S,p\in [P]}c_p\y_{ip} = \sum_u \floor{\alpha_u}\Delta^u$.

We claim that $\y$ satisfies the two conditions of the roundability property. We check Condition~1 first. Let $p \in [P]$ and let $s$ be the index such that $\Delta^s < p \leq \Delta^{s+1}$. Then
\begin{eqnarray}
\textstyle \sum_{q \geq p} \sum_{i \in S_k} \y_{iq} & = & \textstyle \sum_{u \geq s+1} \sum_{i \in S_k} \y_{i\Delta^u} \ = \ \sum_{u \geq s+1} \floor{\alpha_u} \leq \textstyle \ \floor{\sum_{u \geq s+1} \alpha_u } \\
& = &  \textstyle \floor{\sum_{u \geq s+1} \sum_{(i,q) \in A_u} y_{iq}} \ \leq \floor{\sum_{q: c_q \geq \Delta^{s+1}} \sum_{i \in S_k} y_{iq} } \\
& \leq &  \textstyle\floor{\sum_{q \geq p}\sum_{i\in S} y_{iq}}
\end{eqnarray}
%%Note that whenever we set $\y_{i,2^q} = 1$ for all $i\in F_q$, we can charge it to $\sum_{(i,p)\in S_q} y_{i,p}$. That this charging preserves condition 1 is because
where in the second-last inequality  we have used the fact that $c_q \geq \effc(i,q)$ for any $i,q$.\smallskip
%Also note that a similar argument yields $\sum_{i\in S,p\in [P]} c_p \y_{i,p} = \sum_{u=0}^U 2^{u} \floor{\alpha_u/2}$.
%

We now need to prove condition 2 is satisfied.
Call the parameter $u$ {\em good} if $\alpha_u \geq \frac{1}{\epsilon}$ and bad otherwise.
Note that if $u$ is good, then $\alpha_u \leq (1+\epsilon)\floor{\alpha_u}$.
For simplicity, let
$
D:= \sum_{j\in C} d_j \sum_{i\in S_k,p\in [P]} \x_{ijp}
$ denote the total fractional demand assigned to $S_k$. %Claim~\ref{fclm:imp2} shows that $D \geq 16 \cdot 2^U$.
From the definition of $\effc(\cdot)$, we get $\effc(i,p)y_{ip} = \sum_{j\in C} d_j \x_{ijp}$ since for $j\notin H,i\in H$ we have $\x_{ijp} = 0$. Therefore,
\begin{eqnarray}
D := \sum_{j\in C} d_j \sum_{i\in S_k,p\in [P]} \x_{ijp} & = &  \sum_{i\in S_k,p\in [P]} \effc(i,p)y_{i,p}	\  \leq \ \sum_{u\leq U} \Delta^{u+1} \sum_{(i,p)\in A_u} y_{ip} \\
	& \leq & (1/\epsilon)\sum_{u\le U: \textrm{ bad }} \Delta^{u+1} + \sum_{u\le U: \textrm{ good}} \Delta^{u+1} \alpha_u\\
	&	  \leq & \  (1/\epsilon) \cdot \frac{\Delta^{U+2}}{\Delta - 1} + (1+\epsilon)\Delta \sum_{u: \textrm{good}}\Delta^{u} \floor{\alpha_u}  \\
																				& \leq & \frac{(1+\epsilon)^2}{\epsilon^2}\Delta^U + (1+\epsilon)^2 \sum_{i\in S,p\in [P]}c_p\y_{ip} \\
																				& \leq & \epsilon(1+\epsilon)^2 D + (1+\epsilon)^2 \sum_{i\in S,p\in [P]}c_p\y_{ip}
\end{eqnarray}
Therefore, we get $D \leq \frac{(1+\epsilon)^2}{1 - \epsilon(1+\epsilon)^2}\sum_{i\in S,p\in [P]} \y_{ip} \leq (1+100\epsilon)\sum_{i\in S,p\in [P]} \y_{ip} \leq (1+\delta)\sum_{i\in S,p\in [P]} \y_{ip}$.
The second inequality uses $\epsilon$ is small enough.
%Therefore, $S_k$ has the $(20,32)$-roundability property.
Therefore, $S_k$ has the $\left(\tilde{O}(1/\delta),(1+\delta)\right)$-roundability property.
\end{proof}
\noindent
Claim~\ref{fcl:phi-augment}, Lemma~\ref{flem:local}, Claim~\ref{fclm:imp}, and Lemma~\ref{flem:sk} prove that the decomposition has the properties desired by
Theorem~\ref{fthm:decomp}.
\section{Supply Polyhedra of $Q||C_{min}$: Proof of Theorem~\ref{fthm:asslp}}\label{fsec:asslp}
\def\pv{\mathbf{b}}
\newcommand{\dem}{\mathsf{cap}}
  Throughout the proof we fix $\calI$ to be the instance of $Q||C_{min}$ % and we rename machines and jobs so that $D_1 \geq \cdots \geq D_m$ and $n$ types of  jobs $J$ with capacities $c_1 \geq \cdots \geq c_n$.
  %We also fix
  and the supply vector $(s_1,\ldots,s_n)$. For simplicity of presentation, given the supply vector, abusing notation let $J$ denote the multiset of jobs where job $j$ appears $s_j$ times. We know that the LP\eqref{feq:asslp1}-\eqref{feq:asslp3} is feasible with the $s_j$ replaced by $1$. %Let $N = \sum_j s_j$.
  We want to find an assignment where machine $i$ gets at least $D_i/2$ capacity.

The algorithm is a  very simple greedy algorithm which doesn't look at the LP solution , and the feasibility of LP\eqref{feq:asslp1}-\eqref{feq:asslp3} is only used for analysis.
Rename the jobs (with multiplicities) in decreasing order of capacities $c_1\geq c_2 \geq \cdots \geq c_N$, and rename  the machines in decreasing order of $D_i$'s, that is, $D_1 \geq D_2 \geq \ldots \geq D_m$.
Starting with machine $i=1$ and job $j=1$, assign jobs $j$ to $i$ if the total capacity filled in machine $i$ is $< D_i/2$ and move to the next job. Otherwise, call machine $i$ happy and move to the next machine. Obviously, if all machines are happy at the end we have found our assignment.

The non-trivial part is to  prove that if some machine is unhappy, then the LP\eqref{feq:asslp1}-\eqref{feq:asslp3} is infeasible (with $s_j$ replaced by $1$).
To do so, we take the Farkas dual of the LP; the following LP is feasible iff LP\eqref{feq:asslp1}-\eqref{feq:asslp3} is infeasible. We describe a feasible solution to the system below if we obtain some unhappy agent.
	\begin{alignat}{4}
		&&   & \quad \textstyle \sum_{i=1}^m \beta_i D_i > \sum_{j=1}^n\alpha_j \label{feq:assdual1}  \tag{F1} \\
		&& \quad \forall i\in M,j\in J & \quad \textstyle \beta_i\min(c_j,D_i) \leq \alpha_j \label{feq:assdual2} \tag{F2}  \\
		&& \quad \forall i\in M, &\quad  \beta_i \geq 0\label{feq:assdual3}\tag{F3}
	\end{alignat}
	\def\i{i^\star}
Suppose machine $\i$ is the first machine which is unhappy. Let $S_1,\ldots,S_{\i-1}$ be the jobs assigned to machines $1$ to $(\i-1)$ and $S_{i^\star}$ be the remainder of jobs.
We have $\sum_{j\in S_{\i}} c_j < D_{\i}/2$. We also have for all $1\leq i\leq \i$, $\sum_{j\in S_i} \min(c_j,D_i) \leq D_i$. If not, then the machine must receive at least two jobs and would have capacity $> D_i/2$
from all but the last.
We now describe a feasible solution to \eqref{feq:assdual1}-\eqref{feq:assdual3}.
%{\bf deepc: this turns out to be trickier than what meets the eye}.

Given the assignment $S_i$'s, call a machine $i$ {\em overloaded} if $S_i$ contains a single jobs $j_i$ with $c_{j_i} \geq D_i$.
We let $\beta_1 = 1$. For $1\leq i <\i$, we have the following three-pronged rule
\begin{itemize}[noitemsep]
	\item If $i+1$ is not overloaded, $\beta_{i+1} = \beta_i$.
	\item If $i+1$ is overloaded, and so is $i$, then $\beta_{i+1} = \beta_i \cdot D_i/D_{i+1}$.
	\item If $i+1$ is overloaded but $i$ is not, then $\beta_{i+1} = \beta_i \cdot c_{j_{i+1}}/D_{i+1}$, where $j_{i+1}$ is the job assigned to $i+1$.
\end{itemize}
For any job $j$ assigned to machine $i$, we set $\alpha_j = \beta_i \min(c_j,D_i)$. Since for any $S_i$, we have $\sum_{j\in S_i} \min(c_j,D_i) \leq D_i$ and $\sum_{j\in S_{\i}} c_j < D_{\i}/2$,  the given $(\alpha,\beta)$ solution satisfies \eqref{feq:assdual1}. We now prove that it satisfies \eqref{feq:assdual2}.
From the construction of the $\beta$'s the following claims follow.
\begin{claim}\label{fclm:c1}
$\beta_1\leq \beta_2 \leq \cdots \leq \beta_m$.
\end{claim}
\begin{claim}\label{fclm:c2}
$\beta_1D_1 \geq \beta_2D_2 \geq \cdots \geq \beta_mD_m$.
\end{claim}
\begin{proof}
	The only non-obvious case is if $i+1$ is overloaded but $i$ is not: in this case $\beta_{i+1}D_{i+1} = \beta_ic_{j_{i+1}}$. But since $i$ is not overloaded, let $j$ be some job assigned to $i$ with $c_j \leq D_i$.
	By the greedy rule, $c_j \geq c_{j_{i+1}}$, and so $\beta_{i+1}D_{i+1} \leq \beta_iD_i$.
\end{proof}
\noindent
Now fix a job $j$ and let $i$ be the machine it is assigned to. Note \eqref{feq:assdual2} holds for $(i,j)$ and we need to show \eqref{feq:assdual2} holds for all $(i',j)$ too.
I don't see any more glamorous way than case analysis. \smallskip

\noindent
{\bf Case 1: $c_j \leq D_i$.} In this case $\alpha_j = \beta_ic_j$ and $i$ is not overloaded.
Let $i' < i$.  Then we have $\beta_{i'}\min(c_j,D_{i'}) \le \beta_{i'}c_j \leq \beta_ic_j$, where the last inequality follows from Claim~\ref{fclm:c1}.

Now let $i' > i$. If $c_j \leq D_{i'}$, then none of the machines from $i$ to $i'$ can be overloaded. Therefore, $\beta_{i'} = \beta_i$, and so $\beta_{i'}c_j = \beta_ic_j = \alpha_j$.
So, we may assume $c_j > D_{i'}$ and we need to upper bound $\beta_{i'}D_{i'}$. Let $i'' > i$ be the first machine which is overloaded with job $j''$ say.
By Claim~\ref{fclm:c2}, we have $\beta_{i'}D_{i'} \leq \beta_{i''}D_{i''}$. Now note that
$\beta_{i''}D_{i''}  = \beta_{i''-1}c_{j''} = \beta_ic_{j''} \leq \beta_ic_j = \alpha_j$ where the second equality follows since none of the machines from $i$ to $i''-1$ were overloaded. \smallskip

\noindent
{\bf Case 2: $c_j > D_i$.} In this case $\alpha_j = \beta_iD_i$ and $i$ is overloaded. Let $i' > i$. Then, $\beta_{i'}\min(c_j,D_{i'}) = \beta_{i'}D_{i'} \leq \beta_iD_i$ where the last inequality follows from Claim~\ref{fclm:c2}.

Let $i' < i$. Let $i'\leq i'' < i$ be the smallest entry such that $c_j > D_{i''}$. Note that all machines from $i''$ to $i$ must be overloaded implying $\beta_{i''}D_{i''} = \beta_iD_i$.
Since $c_j \leq D_{i'}$ (in case $i' < i''$), we need to upper bound $\beta_{i'}c_j$.
By Claim~\ref{fclm:c1}, $\beta_{i'}c_j \leq \beta_{i''-1}c_j$. Now, if $i''-1$ were overloaded,
then $\beta_{i''}D_{i''} = \beta_{i''-1}D_{i''-1} \geq \beta_{i''-1}c_j$ where the last inequality follows from definition of $i''$. Together, we get $\beta_{i'}c_j \leq \beta_iD_i$.
%For every machine $1\leq i\leq \i$, define $\beta_i := D_\i/D_i$. For every job $j\in S_i$ for $1\leq i\leq \i-1$, define $\alpha_j = \beta_i\min(c_j,D_i)$ and for $j\in S_\i$, define $\alpha_j = c_j$.
%\end{proof}
\def\y{\bar{y}}
\def\z{\bar{z}}
\def\yy{\bar{\bar{y}}}
\begin{lemma}\label{flem:conf-is-uf}
	$\calP_\mathsf{ass}$ is upward-feasible.
\end{lemma}
\begin{proof}
	%\comment{\bf \Large Needs to be written}
	Let $s := (s_1,\ldots,s_n)\in \calP_\mathsf{ass}$ for a certain instance of $Q||C_{min}$ where the jobs have been renamed so that $c_1\leq \cdots \leq c_n$. We need to prove any non-negative vector $t := (t_1,\ldots,t_n)$ s.t. $t\succeq_\suff s$also lies in $\calP_\mathsf{ass}$.
	By the ``hybridization argument'', it suffices to prove the lemma for  $s$ and $t$ differing only in coordinates $\{j-1,j\}$ and $ t_j\ge s_j$ and $t_{j-1} \geq \max(0,s_{j-1} + (s_j - t_j))$.
	Given that, we can move from $s$ to $t$ by changing pairs of coordinates each time maintaining feasibility in $\calP_\mathsf{ass}$.
	
	Let $z$ be the solution for the supply vector $s$; we construct a solution $\z$  for the supply vector $t$ starting with $\z = z$.
	If $\z$ is not already feasible, then it must be because $s_{j-1} \geq \sum_{i\in M} \z_{i,j-1} > t_{j-1}$.
	We select an arbitrary $i\in M$ with $\z_{i,j-1} > 0$ and increase $\z_{ij}$ and decrease $\z_{i,j-1}$ by $\delta$. Since $c_j \geq c_{j-1}$, \eqref{feq:asslp2} remains valid.
	Since the total increase of fractional load of job $j$ is exactly the same as the decrease in that of job $j-1$, and we only need total  decrease $(s_{j-1} -t_{j-1}) \leq t_j - s_j$, at the end we get that $\z$ is feasible wrt supply vector $t$.
	
\end{proof}
%\begin{lemma}
%Suppose $(y_1,\ldots,y_n)\in \calP_\mathsf{ass}$. Let $(\y_1,\ldots,\y_n)$ be a vector such that for all $1\leq i\leq n$, $\sum_{j\leq i} \y_j \geq \sum_{j\leq i} y_i$. Then
%$(\y_1,\ldots,\y_n)\in \calP_{\mathsf{ass}}$.
%\end{lemma}
%\begin{proof}
%By induction, let us assume the lemma is true for all $\y$ with $\y_1 = y_1$ which satisfies the prefix-sum condition.
%Let $\yy$ be the vector with $\yy_1 = y_1$, $\yy_2 = \y_2 + \y_1 - y_1$, and $\yy_i = \y_i$ otherwise.
%Since $\yy\in \calP_\mathsf{ass}$, there is an assignment $z_{ij}$ satisfying \eqref{feq:asslp1}-\eqref{feq:asslp3} with $s_j = \yy_j$.
%We now describe a feasible solution $\z_{ij}$ with $s_j = \y_j$.
%
%Let $\theta := \y_2/\yy_2 \le 1$ since $\y_1 \geq y_1$. Define $\z_{i2} = \theta z_{i2}$ for all $i$, and define $\z_{i1} = z_{i1} + (1-\theta)z_{i2}$.
%For $j=2$, we have $\sum_{i\in M} \z_{i2} = \theta \sum_{i\in M} z_{i2} \leq \theta \yy_2 = \y_2$.
%For $j=1$, we have $\sum_{i\in M} \z_{i1} = \sum_{i\in M} z_{i1} + (1-\theta) \sum_{i\in M} z_{i2} \leq \yy_1 + (1-\theta)\yy_2 = \yy_1 +\yy_2 - \y_2 = \y_1$.
%Since the other $z_{ij}$'s and $\y_j$'s are untouched, $\z_{ij}$ satisfies \eqref{feq:asslp1} with $\y_j$'s.
%
%Now fix a machine $i$. The `increase' in the LHS of \eqref{feq:asslp2} is  $\sum_{j\in J} (\z_{ij} - z_{ij}) \min(c_j,D_i) = (1-\theta)z_{i2}\min(c_1,D_i) -  (1-\theta) z_{i2}\min(c_2,D_i) \geq 0$ since $c_1 \geq c_2$.
%
%\end{proof}
%

\def\Supp{\mathsf{Supp}\xspace}
\newcommand{\barcalS}{\bar{\cal S}\xspace}
\renewcommand{\brp}{{(p)}}
\renewcommand{\br}[1]{{(#1)}}
\renewcommand{\brp}{{(p)}}
\renewcommand{\br}[1]{{(#1)}}
\renewcommand{\bc}{{\bar c}}
\newcommand{\brt}{{(t)}}
\def\cc{\tilde{c}}
\newcommand{\barD}{\bar{D}}
\def\calFr{\calF^{(\alpha,\beta)}}
%\newpage
\section{Supply Polyhedra for \cckp: Proof of Theorem~\ref{fthm:conflp}}\label{fsec:conflp}
\def\z{\bar z}
%\begin{theorem}\label{fthm:conflprounding}
%\end{theorem}
%\begin{proof}
	    Throughout the proof we fix $\calI$ to be the instance of \cckp and the supply vector $(s_1,\ldots,s_n)$.
	    Let  $z$ be a feasible solution to \eqref{feq:conflp1}-\eqref{feq:conflp3}. The proof of Theorem~\ref{fthm:conflp} follows from
	    Lemma~\ref{flem:conf-round}, Lemma~\ref{flem:conf-is-uf}, and Lemma~\ref{flem:conf-so}
	    \begin{lemma}\label{flem:conf-round}
	    	Given $z$, we can  find an of assignment of the $s_j$ jobs of capacity $c_j$  to the machines such that for all $i\in M$
	    	receives a total capapcity $\geq D_i/\alpha$ for $\alpha = O(\log D)$ where $D = D_\mathrm{max}/D_\mathrm{min}$.
	    \end{lemma}
	    \begin{proof}
%	    Our objective is to find an assignment of the $s_j$ jobs of capacity $c_j$ to the $m$ machines so that machine $i$ obtains total capacity $\geq D_i/\alpha$ for $\alpha = O(\log D)$. 	
	    We start by classifying the demands into buckets.
	
	    	\paragraph{Bucketing Demands.} We partition the demands into buckets depending on their requirement values $D_i$. By scaling data, we may assume without loss of generality that $D_\mathrm{min} = 1$.
	    	We say that demand $i$ belongs to \emph{bucket $t$} if $2^{t-1} \leq D_i < 2^t$. We let $B^\brt$ to denote the bucket $t$. The number of buckets $K \leq \log_2 D$.
	    	For any bucket $t$, we round-down all the demands for $i\in B^\brt$; define $\barD_i = 2^{t-1}$ for all $i\in B^\brt$. Note that any $\rho$-approximate feasible solution with respect to $\barD$'s is $2\rho$-approximate with respect to the original $D_i$'s. \smallskip

\noindent
  To this end, we modify the feasible solution $z$ to a solution $\z$  in various stages. Initially $\z \equiv z$.
	    Our modified solution $\z$'s support will not be $\calF_i$; to this end we define $\calFr_i$ for parameters $\alpha,\beta \geq 1$.
	    \begin{definition}
For machine $i$ and parameters $\alpha,\beta > 1$, $\calFr_i$ contains the set $S$ if either (a) $S = \{j\}$ is a singleton with $c_j \geq \frac{\barD_i}{3\log_2 D}$,
or (b) $|S|\leq f_i$, $c_j \leq \alpha \cdot \frac{\barD_i}{3\log_2 D}$, and $\sum_{j\in S} c_j \geq \frac{\barD_i}{\beta}$. 	    We say $\z$ is $(\alpha,\beta)$-feasible if for all $i$, $\z(i,S) > 0$ implies $S\in \calFr_i$.
	    \end{definition}
	    \noindent

	    \medskip
	
		\noindent {\bf Step 1: Partitioning Configurations.}	
		
		We call a job of capacity $c_j$ {\em large} for machine $i$ if $c_j \geq \frac{\barD_i}{3\log_2 D}$, otherwise we call it \emph{small} for machine $i$.
%		For a machine $i$, we define a relaxed collection of feasible sets $\calFr_i$ where $S\in \calFr_i$ if either (a) $S = \{j\}$ and $j$ is large for $i$, or (b) $c_j < \frac{D_i}{3\log_2 D}$ for all $j\in S$, $|S| \leq f_i$, and $\sum_{j\in S} c_j \geq D_i/2$.
%	
%		
%
%%		
%			Let $\Delta$ be a large enough constant.
%				For every machine $i$, let $\calFr_i$ denote the subsets $S\subseteq \Supp$ with $|S|\leq f_i$ but $\sum_{j\in S} c_j \geq D_i/4\Delta^2$.
%				
%		%First we round the quantities $c_p$ and $D_j$ down to nearest power of $\Delta$, and let the rounded quantities be $\bc_p$ and $\barD_j$ respectively.
%	   We call a configuration $S\subseteq \Supp$ large for $i$ if $z(i,S) > 0$ and $S$ contains a large job for $i$. Otherwise it is called small.
%
%Our first step is  to modify $z$ such that its support $z(i,S) > 0$ for only $S\in \calFr_i$ for all $i$.
For every machine $i$, if $z(i,S) > 0$ and $S$ contains any large job $j$ for $i$, then we replace $S$ by $\{j\}$. To be precise, we set $\z(i,\{j\}) = z(i,S)$ and $\z(i,S) = 0$.
We call such singleton configurations {\em large} for $i$; all others are {\em small}. %Note that after this step, $z(i,S) > 0$ only for $S\in \calFr_i$.
Let $\calF^L_i$ be the collection of large configurations for $i$; the rest $\calF^S_i$ being small configurations.
Define $\z^L(i) := \sum_{S\in \calF^L_i} \z(i,S)$ be the total large contribution to $i$, and let $\z^S(i) := \sum_{S\in \calF^S_i} \z(i,S)$ the small contribution.

\begin{claim}\label{fclm:step1}
	After {\bf Step 1}, $\z$ satisfies \eqref{feq:conflp1} and \eqref{feq:conflp2} and $\z$ is $(1,1)$-feasible.
\end{claim}

	We partition the demands into buckets depending on their requirement values $D_i$. By scaling data, we may assume without loss of generality that $D_{min} = 1$.
	We say that demand $i$ belongs to \emph{bucket $t$} if $2^{t-1} \leq D_i < 2^t$. We let $B^\brt$ to denote the bucket $t$. The number of buckets $K \leq \log_2 D$.
%		Note that the number of buckets $K \leq \log_2 D$; this drives the approximation factor.
%We make one observation.
%\begin{claim}\label{fclm:c001}
%	For any $t$, let $i$ and $i'$ be two machines in $B^\brt$ and let $f_i \leq f_{i'}$.
%	Let $z(i,T) > 0$ for some small configuration for $i$.
%Then $T\in \calFr_{i'}$ and $\sum_{j\in T} c_j \geq D_{i'}/2$.
%\end{claim}
%\begin{proof}
%Note that since $z(i,T) > 0$, we have $\sum_{j\in T} c_j \geq D_{i} \geq 2^{t-1} \ge D_{i'}/2$. Furthermore, for any $j\in T$, we have $c_j \leq \frac{D_{i}}{3\log_2 D} \le \frac{2^t}{3\log_2 D}$.
%Therefore any other machine $i'\in B^\brt$, $T$ satisfies two conditions of being in $\calFr_{i'}$.
%Now if $f_{i'} \geq f_i$, we get $|T| \leq f_{i'}$ as well.
%\end{proof}
%Before describing our subroutines, we make a few definitions. All of these are with respect to a solution $z$.

%At this stage, note that $\z^L(i) + \z^S(i) = 1$ for all machines $i$.		
A machine $i$ is called {\em rounded} if $\z(i,S) = 1$ for some set $S$. We let $\calR$ denote the rounded machines.
The remaining machines are of three kinds:  {\em large} ones with $\z^L(i) = 1$, {\em hybrid} ones with $\z^L(i) \in (0,1)$ and {\em small} ones with $\z^L(i) = 0$.
Let $\calL,\calH,\calS$ denote these respectively.  \medskip

\noindent{\bf Step 2: Taking care of large machines.}

The goal of this step is to modify $\z$ such that (a) the set of large machines becomes empty and (b) the set of hybrid machines is bounded. In particular, we will have at most one hybrid machine in a bucket
proving there are at most $K$ hybrid machines. First we need to discuss two sub-routines.

\paragraph{Subroutine: {\sf FixLargeMachine}($i$).}
This takes input a large machine $i\in\calL$, that is,  $\z^L(i) = 1$. We modify $\z$ such that at the end of the subroutine, among other things, $i$ gets rounded and enters $\calR$.

Consider the jobs $j$ large for $i$ such that $\z(i,\{j\}) \in (0,1)$. Since $\z^L(i) = 1$ and $i\notin \calR$, there exists at least two such jobs.
%
%Since $i\notin \calR$ in the beginning and yet $\z^L(i) = 1$, there must exist at least two large configurations $z(i,\{j\}) \in (0,1)$. % and $z(i,\{j_2\}) \in (0,1)$.
Let $j_1$ be the smallest capacity among these, and $j_2$  be any other such job. %all large configurations $(i,\{j\})$ with $z(i,\{j\}) > 0$.
%	Wlog, assume $c_{j_1} \le c_{j_2}$.
Two cases arise. In the simple case, there exists no $i'\notin \calR, S'\subseteq \Supp$ with $\z(i',S) > 0$ and $j_1 \in S$. That is, no other machine fractionally claims the job $j_1$.
Since $s_{j_1}$ is an integer, we have slack in \eqref{feq:conflp2}. We round up $\z(i,\{j_1\}) = 1$, set $\z(i,T) = 0$ for all other configurations of $i$, and add $i$ to $\calR$ and terminate.
%
%(zeroing out all other $i$'s $z(i,S)$'s ) without violating \eqref{feq:conflp2}. We then  add $(i,\{j_1\})$ to $\calR$.

Otherwise, there exists a machine $i'$
%(which could be in a different bucket)
and a set $S$ such that $z(i',S) \in (0,1)$ and $j_1 \in S$.
Now define the set $T$ as follows. If $c_{j_2} > \frac{\barD_{i'}}{3\log_2 D}$, then $T = \{j_2\}$; otherwise $T = S - j_1 + j_2$.
Note that in the second case $j_2$ could already be in $S$; $T$ then contains one more copy, that is, $n(T,j_2) = n(S,j_2) + 1$.
%Note that in either case $T \in \calFr_{i'}$.
%In the first case, $j_2$ is large for $i'$. In the second case, $|T| = |S|$ and $\sum_{j\in T} c_j \ge \sum_{j\in S} c_j$ since $c_{j_2} \geq c_{j_1}$ by choice of $j_1$.
We modify $\z$-as follows. We decrease $\z(i,\{j_2\})$ and $\z(i',S)$ by $\delta$, and increase $\z(i,\{j_1\})$ and $\z(i',T)$ by $\delta$ till one of the values becomes $0$ or $1$.
%As before, this preserves the LHS of \eqref{feq:conflp1} and can only decrease the LHS of \eqref{feq:conflp2} (for jobs $j\in S\setminus j_1$ if $T = \{j_2\}$).
%	Also note that $z^L(j)$ can only increase for any job; in particular no job leaves $\calL$. {\Huge NOT CORRECT}
%This process ends with assigning $z(i,\{j_1\}) = 1$ and we add $(i,\{j_1\})$ to $\calR$.
%If $\z(i',T)$ becomes $1$, we add $i'$ to $\calR$
If at any point, some configuration gets $\z$ value $1$, we add the corresponding machine to $\calR$.
We proceed till $i$ enters $\calR$.

\begin{claim}\label{fclm:002}
	{\sf FixLargeMachine}($i$) terminates. Upon termination, the solution $\z$ satisfies \eqref{feq:conflp1} and \eqref{feq:conflp2}, and
	if $\z$ was $(\alpha,\beta)$-feasible before the subroutine, it remains $(\alpha,\beta)$-feasible afterwards.
%	Furthermore, $i$ enters $\calR$.
\end{claim}
\begin{proof}
	If at any point we are in the simpler case, then $i$ enters $\calR$ and we terminate. Since we modify $\z(i,S)$ only for machine $i$, \eqref{feq:conflp1} is satisfied by the modification.
	\eqref{feq:conflp2} is satisfied for $j$ no other machine fractionally claims it. In the other case, note that the modification by $\delta$'s preserve the LHS of \eqref{feq:conflp1}. Furthermore, since $T\subseteq S\cup j_2$, it can only
	decrease the LHS of \eqref{feq:conflp2} (for jobs $j' \in S\setminus T\cup j_1$ when $T=\{j_2\}$ ). Finally, the new entry to the support of $\z$ is $\z(i',T)$ and we need to check $T\in \calFr_i$.
	If $T$ is a singleton (that is $j_2$), then $c_{j_2} \geq \frac{\barD_{i'}}{3\log D}$ and so $T\in \calFr_i$.
  Otherwise, since $S\in \calFr_i$, $c_{j_2} <  \frac{\barD_{i'}}{3\log D}$, and $c_{j_2} \geq c_{j_1}$ we have $T\in \calFr_i$.
  So at every step $\z$ maintains \eqref{feq:conflp1} and \eqref{feq:conflp2} and is $(\alpha,\beta)$-feasible. To argue termination, note that in the second case the value of $\z(i,\{j_1\})$ strictly goes up.
  In the end, we must have $\z(i,\{j_1\}) = 1$.
\end{proof}

\paragraph{Subroutine: {\sf FixBucket}($t$).} This takes input a bucket $t$ with more than one hybrid machine, and modifies the $\z$-solution such that
there is at most one hybrid machine in $t$. Recall a machine is hybrid if $\z^L(i) \in (0,1)$.
The $\z$-value for other machines in other buckets are unaffected.

%Let
%	{\bf Step 2: Rounding Large Demand Assignments.} In this step, we modify the LP solution
% such that for each bucket $t$, there is at most one hybrid machine $i\in B^{(t)}$ with  $z^L(i) \in (0,1)$.
%% and at most one $j$ with $c_j \geq D_i/4\Delta^2$  with $z(i,\{j\}) \in (0,1)$.
%% The flip side is that we may introduce variables $z(i,S)$ for sets $S$ where $|S|\leq f_i$ but $\sum_{j\in S} c_j \geq D_i/4\Delta^2$.
% %\emph{at most one strictly fractional variable} $0 < y_{S,j} < 1$ over all $j \in B^\brt$ and $S \in \barcalS^L_j$.
% To this end, we repeatedly perform the following steps, starting from the smallest bucket $t$ onwards.

%
% Before we start the steps for bucket $t$, we always maintain the following invariant holds for every bucket $t' < t$ (which is vacuously true for the smallest ($t=1$) bucket):
%	\begin{framed}
%		\begin{itemize}
%			\item[({\bf I})] For each bucket $t' < t$, there is at most one  demand $i_{t'} \in B^{(t')}$ %\setminus \calR$
%			such that $z^L(i_{t'}) \in (0,1)$. Further, if such a demand $i_{t'}$ exists, then $f_{i_{t'}} \leq f_{i}$ for all $i \in B^{(t')} \setminus \calR$.
%		\end{itemize}
%	\end{framed}
	
%	Suppose this invariant holds for all buckets upto $t-1$. Now we describe the iteration for bucket $t$.
	
Among the hybrid machines in $B^\brt$, let $i$  be the one with the smallest $f_i$. Let $i'$ be any other hybrid machine in this bucket. We know there is at least one more.
%	
%	
%	%\medskip \noindent {\bf Step 2a: Intra-Bucket Rounding.}
%	Define a total order $\prec$ on $B^\brt \setminus \calR$, where $i \prec i'$ if $f_{i} \leq f_{i'}$.
%	%We now ensure that the demands with the smallest $f_j$ value get preference when it comes to being satisfied by large configurations.
%	Suppose there exists $i,i' \in B^\brt \setminus \calR$ such that the following conditions hold: (i) $i\prec i'$, (ii) $z^L(i)$ and $z^L(i')$ are both in $(0,1)$.
	We now \emph{modify} $\z$ as follows.
	Since $\z^L(i') > 0$, there exists a large configuration $\{j'\}$ for $i'$ with $\z(i',\{j'\}) > 0$. Similarly, since $\z^L(i) < 1$, there must exist a {\em small} configuration $T$ with $\z(i,T) > 0$.
	%To this end, let $(S,j') \in \barcalS^L_{j'}$ be a large configuration  with $y(S,j') > 0$. Now, since $y^l(j) < 1$, we know that there is a small configuration $(T,j) \in \barcalS^S_j$ with $y(T,j) > 0$.
%	By Claim~\ref{fclm:c001}, note that $T\in \calFr_{i'}$ as well.
	We then perform the following change: decrease $\z(i',\{j'\})$ and $\z(i,T)$ by $\delta$, and increase $\z(i,\{j'\})$ and $\z(i',T)$ by $\delta$,  for a $\delta > 0$ such that one of the variables becomes 0 or 1.
	Note that this keeps \eqref{feq:conflp1} and \eqref{feq:conflp2} maintained. 	%In particular, no job $j$ leaves $\calL$.
	
	We keep performing the above step till bucket $t$ contains at most one hybrid machine.
	If at any point, some configuration gets $\z$ value $1$, we add the corresponding machine to $\calR$.
\begin{claim}\label{fclm:step2}
	{\sf FixBucket}($t$) terminates. Upon termination, the solution $\z$ satisfies \eqref{feq:conflp1} and \eqref{feq:conflp2}, and
	if $\z$ was $(\alpha,\beta)$-feasible before the subroutine, it remains $(\alpha,\beta)$-feasible afterwards.
\end{claim}
	\begin{proof}
The possibly new entry to the support of $\z$ is $\z(i',T)$. Note that $|T| \leq f_i$ since $\z$ was $(\alpha,\beta)$-feasible to begin with,  and therefore $|T|\leq f_{i'}$ as well.
The other conditions of $(\alpha,\beta)$-feasibility are satisfied since $\barD_i = \barD_{i'}$, both being in the same bucket.
Also note that the LHS of both \eqref{feq:conflp1} and \eqref{feq:conflp2} remain unchanged.
To argue termination, till bucket $t$ contains more than one hybrid machine, note that $\z^L(i)$ increases for the hybrid machine $i$ with the smallest $f_i$.
	\end{proof}

\noindent
Now we have the two subroutines to describe {\bf Step 2} of the algorithm. It is the following while loop.
   \begin{itemize}[noitemsep]
   	\item[{}] While $\calL$ is non-empty:
   	\begin{itemize}[noitemsep]
   		\item If $i\in \calL$, then {\sf FixLargeMachine}($i$). Note that $i$ enters $\calR$ after this. This can increase the number of hybrid machines across buckets.
   		\item For all $1\leq t\leq K$, if $B^\brt$ contains more than one hybrid machine, then {\sf FixBucket}($t$). This can increase the number of machines in $\calL$.
   	\end{itemize}
   \end{itemize}

%   If there are more than one hybrid in any bucket, we apply Step 2 again. Since Step 3 moves machines to $\calR$, this process can go on for at most $m$ steps. After doing this sequence of steps, we have the following scenario.
%
	  Since the {\sf FixLargeMachine} adds a new machine to $\calR$, it cannot run more than $m$ times. Therefore, the while loop terminates. Furthermore, before the loop $\z$ is $(1,1)$-feasible satisfying \eqref{feq:conflp1} and \eqref{feq:conflp2} (Claim~\ref{fclm:step1}), therefore Claim~\ref{fclm:002} and Claim~\ref{fclm:step2} imply that it satisfies after the while loop. We encapsulate the above discussion in the following claim about Step 2.
   \begin{claim}\label{fclm:003}
   	{\bf Step 2} terminates. Upon termination, the modified LP solution $\z$ is $(1,1)$-feasible, satisfies \eqref{feq:conflp1} and \eqref{feq:conflp2}, and furthermore
   	$\calL$ is empty and for every bucket $t$ we have at most one hybrid machine $i\in B^\brt\setminus\calR$.
%   	At the end of {\bf Step 1}, we have for every bucket $t$, at most one  $i\in B^\brt\setminus\calR$ has $z^L(i) \in (0,1)$ and the rest have $z^S(i) = 1$. Furthermore, for every $i\in \calS$ and $z(i,S) > 0$,
%   	we have $\sum_{j\in S} c_j \geq D_i/2$.
   	\end{claim}
   	\smallskip

\noindent{\bf Step 3: Taking care of hybrid machines.}

Let $\calH$ be the set of hybrid machines at this point. We know that $|\calH| \leq K \leq \log_2 D$ since each bucket has at most one hybrid machine.
For any machine $i\in \calH$ with $\z^L(i) \le 1-1/K$, we zero-out all its large contribution. More precisely, for all $j$ large for $i$ we set $\z(i,\{j\})= 0$.
Note that \eqref{feq:conflp1} no longer holds, but it holds with RHS $\geq 1/K$. Note that these machines now leave $\calH$ and enter $\calS$.

At this point, for every $i\in \calH$ has $z^L(i) > 1-1/K$. Let $K' := |\calH|$. Let $J'$ be the set of jobs $j$ which are large for some machine $i\in \calH$ and $\z(i,\{j\}) > 0$.
Let $s'_j := s_j - \sum_{i\in \calR} \sum_S \z(i,S)n(S,j)$ be the remaining copies of $j$. Note that it is an integer since $s_j$ was an integer and for all $i\in \calR$, $\z(i,S) \in \{0,1\}$.
Let $G$ be a bipartite graph with $\calH$ on one side and $J'$ on the other with $s'_j$ copies of job $j$. We draw an edge $(i,j)$ iff $j$ is large for $i$ with $\z(i,\{j\}) > 0$.
\begin{claim}
	There is a matching in $G$ matching all $i\in \calH$.
\end{claim}
\begin{proof}
Pick a subset $\calH' \subseteq \calH$ and let $J''$ be its neighborhood in $G$. We need to show $\sum_{j\in J''} s'_j \geq \calH'$.
Since $z$ satisfies \eqref{feq:conflp2}, we get
\[
\sum_{j\in J''} s'_j \geq \sum_{j\in J''} \sum_{i\in \calH'} z(i,\{j\}) = \sum_{i\in \calH'} \sum_{j\in J''} z(i,\{j\})  > (1-1/K)|\calH'| \geq |\calH'| - 1
\]
The first inequality follows since $\z$ satisfies \eqref{feq:conflp2}.
The strict inequality follows since $J''$ is the neighborhood of $\calH'$ and the fact that $z^L(i) > 1-1/K$ for all $i\in \calH$.
The claim follows since $s'_j$'s are integers.
\end{proof}
If machine $i\in \calH$ is matched to job $j$, then we assign $i$ a copy of this job, that is,  set $\z(i,\{j\}) = 1$ and $\z(i,S) = 0$ for all other $S$,
and add $i$ to $\calR$. Let $J_M \subseteq J'$ be the sub(multi)set of jobs allocated; note $|J_M| \leq K \leq \log_2 D$.
After this point all machines outside $\calR$ are small. For every $i\in \calS$ and every small configuration $S$ with $\z(i,S) > 0$, we move this mass to $\z(i,S\setminus J_M)$.
More precisely, $\z(i,S\setminus J_M) = \z(i,S)$ and $\z(i,S) = 0$ for all $i$ and $S$. Note that \eqref{feq:conflp2} is satisfied at this point. Furthermore,
since $\z$ was $(1,1)$-feasible, we know that $\sum_{j\in S} c_j \geq \barD_i$ and for every $j\in S\cap J_M$ we have $c_j \leq \frac{\barD_i}{3\log_2 D}$.
\[
\sum_{j\in S\setminus J_M} c_j \geq  \sum_{j\in S} c_j - |J_M|\cdot \frac{\barD_i}{3\log_2 D} \geq \frac{2\barD_i}{3}
\]
Therefore, we have proved the following claim.
\begin{claim}\label{fclm:007}
At the end of {\bf  Step 3}, we have a solution $\z$ with (a) $\z^L(i) = 0$ for all $i\notin \calR$, (b) $\z$ is $(1,3/2)$-feasible,
(c) $\z$ satisfies \eqref{feq:conflp2}, and satisfies \eqref{feq:conflp1} replaced by $\frac{1}{K} \leq \sum_S \z(i,S) \leq 1$.
%
%we have a set of residual machines $\calS$ and a set of residual jobs $J_{res}$  and a solution $z(i,S)$ where
%\begin{enumerate} [noitemsep]
%	\item For all $i\in \calS$ we have $z(i,S) > 0$ iff $|S| \leq f_i$, $\sum_{j\in S} c_j \geq 3D_i/8$, and $c_j < \frac{D_i}{3\log_2 D}$ for all $j\in S$.
%	\item $\forall i \in \calS, ~ \textstyle 1 \ge \sum_{S} z(i,S)  \geq   1/K \geq \frac{1}{\log_2 D}$.
%	\item $\forall j\in J_{res}, ~ \textstyle \sum_{i} z(i,S)n(S,j)  \leq  s_j$.
%	
%\end{enumerate}
\end{claim}
\smallskip

\noindent
{\bf Step 4: Taking care of Small Machines.}
\def\zz{z^{\mathsf{int}}}
\def\2z{\mathsf z}
We now convert the solution $\z$  to a solution $\2z$ of the assignment LP in the following standard way.
As before, let $s'_j = s_j - \sum_{i\in \calR}\sum_S \z(i,S)n(S,j)$ be the number of jobs remaining.
For every $i\notin \calR$ and $j\in J$ define $\2z_{ij} = \sum_{S}\z(i,S)n(S,j)$.  Note that this satisfies the constraint of the assignment LP:
	\begin{alignat}{4}
		&& \quad \forall j \in J,   &\quad  \textstyle \sum_{i\in \calS} \2z_{ij}  \leq  s'_j \label{feq:1} \\
		&& \quad \forall i\in \calS ,      &\quad  \textstyle \sum_{j\in J}  \2z_{ij}c_j \geq \frac{2\barD_i}{3\log_2 D} \label{feq:2} \\
	&& \quad \forall i\in \calS ,      &\quad  \textstyle \sum_{j\in J}  \2z_{ij} \leq f_i \label{feq:3} \\
		&& \quad \forall i\in \calS, j\in J~\textrm{with}~ \textstyle c_j \geq \frac{\barD_i}{3\log_2 D}, & \quad \2z_{ij}   =  0  \label{feq:4}
	\end{alignat}
	The last equality follows since $\z$ was $(1,8/15)$-feasible and so $\z(i,S) = 0$ for any set $S$ containing a job $j$ with $c_j \geq \frac{\barD_i}{3\log_2 D}$.
The first inequality follows  since $\z$ satisfies \eqref{feq:conflp2}. To see the second and third point, note
that for any $i\in \calS$,
\[
\sum_{j\in J} \2z_{ij}c_j = \sum_j \sum_S \z(i,S)n(S,j)c_j = \sum_S \z(i,S) \sum_j n(S,j)c_j \geq \frac{1}{\log_2 D}\cdot \frac{2\barD_i}{3}
\]
	since $\sum_S \z(i,S) \geq 1/K$ for all $i\in \calS$ and since $\z$ is $(1,3/2)$-feasible, we have $\sum_{j=1}^n n(S,j) c_j \geq \frac{2\barD_i}{3}$. Similarly,
\[
\sum_{j\in J} \2z_{ij} = \sum_j \sum_S \z(i,S)n(S,j) = \sum_S \z(i,S) \sum_j n(S,j) \leq f_i
\]
since for any $S$, $\sum_{j\in S} n(S,j) \leq f_i$ and $\sum_S \z(i,S) \leq 1$.
Now we use Theorem~\ref{fthm:shmoystardos} to find an integral allocation $\zz$ of the jobs $J$ to machines in $\calS$ satisfying \eqref{feq:1},\eqref{feq:3}, and $\sum_{j\in J} \zz_{ij}c_j \geq \frac{\barD_i}{3
\log_2 D}$. \medskip

The final integral assignment is as follows. For every $i\in \calR$, we assign the configuration $S$ with $\z(i,S) = 1$.
Since $\z$ is $(1,3/2)$-feasible, every such machine $i$ gets a total capacity of at least $\frac{\barD_i}{3\log_2 D}$. % \geq \frac{\barD_i}{16\log_2 D}$.
All the remaining machines $i\in \calS$ obtain a set of jobs giving them capacity $\geq \frac{\barD_i}{3\log_2 D}$. %\geq \frac{3D_i}{3\log_2 D}$.
This completes the proof of Lemma~\ref{flem:conf-round}.
\end{proof}

\begin{lemma}\label{flem:conf-is-uf}
$\calP_\mathsf{conf}$ is upward-feasible.
\end{lemma}
\begin{proof}
%\comment{\bf \Large Needs to be written}
Let $s := (s_1,\ldots,s_n)\in \calP_\mathsf{conf}$ for a certain instance of \cckp where the jobs have been renamed so that $c_1\leq \cdots \leq c_n$. We need to prove any non-negative vector $t := (t_1,\ldots,t_n)$ s.t. $t\succeq_\suff s$also lies in $\calP_\mathsf{conf}$.
By the ``hybridization argument'',it suffices to prove the lemma for  $s$ and $t$ differing only in coordinates $\{j-1,j\}$ and $ t_j\ge s_j$ and $t_{j-1} \geq \max(0,s_{j-1} + (s_j - t_j))$.
Given that, we can move from $s$ to $t$ by changing pairs of coordinates each time maintaining feasibility in $\calP_\mathsf{conf}$.

Let $z$ be the solution for the supply vector $s$; we construct a solution $\z$  for the supply vector $t$ starting with $\z = z$.
If $\z$ is not already feasible, then it must be because $s_{j-1} \ge \sum_{i,S} z(i,S)n(S,j-1) > t_{j-1}$.
Therefore, we need to decrease the fractional utilization of job $(j-1)$ by $s_{j-1} - t_{j-1} \leq t_j - s_j$.
For any machine $i$ and any set $S\in \calF_i$ with $z(i,S) > 0$ and $n(S,j-1) \geq 1$ (and this must exist since $t_{j-1}\geq 0$),
define $T := S - \{j-1\} + \{j\}$. Note that $T$ could already have a copy of job $j$; we have $n(T,j) = n(S,j) + 1$. Also note since $c_j \geq c_{j-1}$, if $S\in \calF_i$ then so is $T\in \calF_i$.
We let $\z(i,S) = z(i,S) - \delta$ and $\z(i,T) = z(i,T) + \delta$ till
either $\z(i,S) = 0$ or $\z(i,T) = 1$. Since the total increase of fractional load of job $j$ is exactly the same as the decrease in that of job $j-1$, and we only need total  decrease $(s_{j-1} -t_{j-1}) \leq t_j - s_j$, at the end we get that $\z$ is feasible wrt supply vector $t$.
\end{proof}
\begin{lemma}\label{flem:conf-so}
	$\calP_\mathsf{conf}$ has an $(1+\epsilon)$-approximate separation oracle.
\end{lemma}
\begin{proof}
	Fix $\epsilon> 0$.
	Given a supply vector $s = (s_1,\ldots,s_n)$, we give a polynomial time algorithm which either returns a hyperplane separating $s$ and $\calP_\mathsf{conf}$, or we can assert that
	$s\in \calP_\mathsf{conf}(\calI')$, where $\calI'$ is an instance where machine $i$ has demand $D_i/(1+\epsilon)$. To this end, for every machine $i$, define $\calF^{(\epsilon)}_i := \{S: |S| \leq f_i,  \sum_j c_jn(S,j)\geq D_i/(1+\epsilon) \}$. To prove $s\in \calP_\mathsf{conf}(\calI')$, we need to find $z(i,S)$ defined for all $i\in M, S\in \calF^{(\epsilon)}_i$ satisfying \eqref{feq:conflp1}-\eqref{feq:conflp2}.
	For every $j\in J$, define $\cc_j := (1+\epsilon)c_j$. Note for every $S\in \calF^{(\epsilon)}_i$ iff $|S| \leq f_i$ and  $\sum_{j\in S} \cc_j n(S,j) \geq D_i$.
	
%Let $\cC$ be empty.
Consider the following system of inequalities.
%In this section we show how to solve the LP relaxation given in Section~\ref{fsec:lp}. The set of constraints specify a feasible region for the set of variables $(x_{ijp}, y_{ip})$, which is convex. Therefore, we can invoke the ellipsoid algorithm to solve the LP as long as we can find a separating hyperplane for any infeasible solution. So let $(x_{ijp}, y_{ip})$ be a tentative solution. Since the  constraints in~\eqref{feq:mkc1-full}-\eqref{feq:mkc3-full} are polynomial in number, we can check feasibility for these constraints easily.
%
%
%Consider a {\em fixed} set  of clients $J$.
%We show how to check the constraints~\eqref{feq:proj} for $J$. We re-write the constraints~\eqref{feq:conflpnew} for $J$ below:
%\begin{eqnarray*}
%	\sum_{S \in \calF_J} z_{S,J} & \geq & 1 \\
%	\sum_{S\in \calF_{J}} z_{S,J}\cdot n(S,p) &\leq  & \sum_{i\in \Gamma(J)} y_{ip}  \ \ \ \forall p \in [P] \\
%	z_{S,J} & \geq &  0 \ \ \ \forall S \in \calF_J
%\end{eqnarray*}
%
%Note that we do not explicitly need to say $z_{S,J} \leq 1$ as existence of a solution to the above constraints also guarantees another solution
%which satisfies $z_{S,J} \leq 1$ for all $S \in \calF_J$. Treating $y_{ip}$ values as constants, we know by Farkas' lemma that either there is a feasible solution to
%the above constraints, or there is a feasible solution to the following set of constraints (where the variables are $\alpha_p, p \in [P]$),
%which we call dual LP:
\begin{alignat}{4}
	& \quad \forall j \in J,   &&\quad  \textstyle \alpha_j  \geq 0 \label{feq:d3} \tag{D1} \\
	& \quad  &&\quad  \textstyle \sum_{j\in J}  s_j\cdot \alpha_j  < \sum_{i\in M}\beta_i \label{feq:d1} \tag{D2}\\
	& \quad \forall i\in M, S\in\calF_i, && \quad \textstyle \sum_{j\in J} \alpha_j n(S,j)  \geq \beta_i  \label{feq:d2}  \tag{D3}
\end{alignat}
%\begin{eqnarray}
%\textstyle \sum_{j \in J} s_j \cdot  \alpha_j & < & \textstyle \sum_{i\in M} \beta_i
%\label{feq:d1}\\
%\textstyle \sum_{j \in J} n(S,j) \cdot \alpha_j & \geq & \beta_i \ \ \ \forall i\in M, \forall S \in \calF_i
%\label{feq:d2}\\
%\textstyle \alpha_j & \geq & 0 \label{feq:d3} % \tag{D3}
%\end{eqnarray}
We also need a stronger set of 	inequalities.
\begin{equation}\label{feq:d4}
 \textstyle \forall i\in M, S\in\calF^{(\epsilon)}_i, \quad \sum_{j\in J} \alpha_j n(S,j)  \geq \beta_i  \tag{D4}
\end{equation}
If there exists a feasible solution $(\alpha,\beta)$ to \eqref{feq:d3}-\eqref{feq:d2}, then this forms the hyperplane separating $s$ and $\calP_\mathsf{conf}$ as follows.
This is because for all $t\in \calP_\mathsf{conf}$, if  $z(i,S)$ is the solution feasible for $\calP_\mathsf{conf}$ with $t_j$'s in the RHS of \eqref{feq:conflp2},
then $\sum_{i\in M}\beta_i = \sum_{i\in M} \sum_{S\in \calF_i} \beta_i z(i,S) \leq \sum_{i\in M,S\in \calF_i} z(i,S) \sum_{j\in J} \alpha_jn(S,j) \leq \sum_{j\in J} \alpha_j t_j$.
The following claim proves the lemma.\smallskip

\begin{claim}
In polynomial time, we can either find $(\alpha,\beta)$ feasible for \eqref{feq:d3}-\eqref{feq:d2}, or we can find variables $z(i,S)$ for $i\in M,S\in \calF^{(\epsilon)}_i$ satisfying
\eqref{feq:conflp1}-\eqref{feq:conflp2}.
\end{claim}
\begin{proof}
We run the ellipsoid algorithm to check feasibility of the stronger system \eqref{feq:d3},\eqref{feq:d1}, and \eqref{feq:d4}.
At any point, we have a running iterate $(\alpha,\beta)$. %The non-trivial part is to find a separation oracle for \eqref{feq:d2}.
%That is, we need to find $i\in M$, and  a set $S\in \calF_i$ such that $\sum_j \alpha_jn(S,j) < \beta_i$. Unfortunately this is an NP-hard problem and we need to resort to approximation as follows.
For every $i\in M$, maximize $\sum_j \cc_j n(S,j)$ over all subsets $S$ with $|S|\leq f_i$ and $\sum_{j\in J} \alpha_j n(S,j) < \beta_i$.
There is an FPTAS for this problem~\cite{CapraraKPP00}. If the maximum value returned by the approximation scheme is {\em smaller} than $D_i$, then
we know that the true optimum is $\leq D_i(1+\epsilon)$. That is, for every $S$ with $|S| \leq f_i$ and $\sum_{j\in J}\alpha_j n(S,j) < \beta_i$, we have $\sum_{j\in J} \cc_jn(S,j) \le D_i(1+\epsilon)$.
Which in turn implies $\sum_{j\in J} c_jn(S,j) \leq D_i$. Contrapositively, for every $S\in \calF_i$, we must have $\sum_{j\in J} \alpha_j n(S,j) \geq \beta_i$. %implying $(\alpha,\beta)$.
That is $(\alpha,\beta)$ satisfies \eqref{feq:d3}-\eqref{feq:d2} and we exit.
%satisfies \eqref{feq:d1}-\eqref{feq:d3} and we can exit.

Otherwise, the PTAS  returns a set $S^\star$ with $|S^\star|\le f_i$ and $\sum_{j\in J}\cc_j n(S,j) \geq D_i$, that is $S^\star \in \calF^{(\epsilon)}_i$, for which
$\sum_{j\in J}\alpha_j n(S^\star,j) < \beta_i$. We add $(i,S^\star)$ to $\cC$, and return $(\alpha,\beta)$ to the separation oracle for \eqref{feq:d4}.
%We also add $(i,S^\star)$ to $\calC$.
The ellipsoid algorithm states than in polynomial time we either find an $(\alpha,\beta)$ feasible for \eqref{feq:d3}-\eqref{feq:d2}, or the polynomially many hyperplanes
in $\calC$ prove \eqref{feq:d3},\eqref{feq:d1}, and \eqref{feq:d4} is infeasible.
More precisely, there exists a solution $z$ satisfying \eqref{feq:conflp1}-\eqref{feq:conflp2} with $z(i,S)$ defined for $(i,S)\in \calC$.
Since $|\cC|$ is bounded by a polynomial, we can explicitly find $z$ by solving the LP \eqref{feq:conflp1}-\eqref{feq:conflp2} with variables $z(i,S)$ for $(i,S)\in \calC$.
\end{proof}\end{proof}

\subsection{Integrality Gap}\label{fsec:conf-ig}
In this section we prove Theorem~\ref{fthm:conf-ig}.
%We show that the configuration LP for the problem of allocating jobs to machines of different speeds to maximize the minimum processing time has super-constant integrality gap.
%\begin{theorem}
%	The integrality gap of CLP is $\Omega(\log n/\log\log n)$.
%\end{theorem}
%\begin{proof}
	\def\M{\mathcal{M}}
	Fix $K$. We present an instance $\calI_K$ for which configuration LP is feasible but any integral allocation must violate the demand of some machine by factor $K$.
	
	First we describe the machines in $\calI_K$.
	\begin{enumerate}
		\item There is $1$ machine $M_0$ with $D(M_0) = 1$ and $f(M_0) = 1$.
		\item There are $K$ machines $M_1,\ldots,M_K$ with $D(M_i) = K^{-i}$ and $f_i := f(M_i) = K^{2K + 1}\cdot K^{-2i}$.
		\item There are $K$ {\bf classes} of machines $\M_1,\M _2,\ldots, \M _K$. Machines in the same class are equivalent.
		There are $f_i$ machines in $\M_i$ and they are numbered $N^{(i)}_1,\ldots,N^{(i)}_{f_i}$.
		Each machine $N$  in class $i$ has $D(N) = \frac{1}{f_iK^i} = K^{-(2K+ 1)}\cdot K^i$ and $f(N) = 1$.
	\end{enumerate}
	Now we describe the jobs.
	\begin{enumerate}
		\item There are $K$ ``big jobs'' $J_1,\ldots,J_K$ with $c(J_i) = 1$.
		\item There are $K$ other types of jobs of the same capacity. Job $J$ of type $i$ has capacity $c(J) = c_i :=  \frac{1}{f_iK^i} = K^{-(2K+1)}K^i$  and there are $n_i := f_i (1+1/K) = (K+1)K^{2K}K^{-2i}$ of them.
		We divide these $n_i$ jobs into two sets $S_i \cup T_i$ where $|S_i| = f_i$ and $|T_i| = f_i/K$. We order the jobs in $S_i$ arbitrarily and call them $P^{(i)}_1,\cdots,P^{(i)}_{f_i}$.
	\end{enumerate}
	So, the total number of machines in $\calI_K$ are $1 + K + \sum_{i=1}^K f_i  \leq K^{2K}$ and the number of jobs is $K + (1+1/K)\sum_{i=1}^K f_i \approx K^{2K}$.
	
	\begin{lemma}
		The Configuration LP is feasible.
	\end{lemma}
	\begin{proof}
		We describe a fractional solution.
		\begin{enumerate}
			\item For machine $M_0$ we satisfy as follows: set  $y(M_0,J_i) = 1/K$ for $i=1,\ldots,K$. Note $c(J_i) \geq D(M_0)$ and $|J_i| = 1 = f(M_0)$.
			\item For machine $M_i$ we satisfy as follows: set $y(M_i,J_i) = 1-1/K$ and $y(M_i,S_i) = 1/K$. Recall $S_i$ are the $f_i$ jobs of type $i$.
			\begin{itemize}
				\item Note $c(J_i) = 1 \geq D(M_i) = K^{-i}$ and 	$|J_i| = 1 \leq f(M_i) = K^{2K+1}K^{-2i}$ since $i\leq K$.
				\item Note $c(S_i) = |S_i|\cdot \frac{1}{f_iK^i} = \frac{1}{K^i} = D(M_i)$. Note $|S_i| = f_i = f(M_i)$.
			\end{itemize}
			\item For $1\leq i\leq K$, for a machine $N^{(i)}_j$ in class $i$, where $1\leq j\leq f_i$, we satisfy it as follows: $y(N^{(i)}_j, P^{(i)}_j) = 1-1/K$ and $y(N^{(i)}_j, t) = 1/f_i$ for all $t\in T_i$.
			Since $|T_i| = f_i/K$, the total fractional $y$-amount that $N^{(i)}_j$ gets is $1$. Also note that $N^{(i)}_j$ gets singleton jobs of type $i$ whose capacity is $\frac{1}{f_iK_i} = D(N^{(i)}_j)$.
		\end{enumerate}
		We need to show that no job is over allocated.
		\begin{enumerate}
			\item The big jobs $J_i$ is given $1/K$ to $M_0$ and $(1-1/K)$ to $M_i$.
			\item For $1\leq i\leq K$, $1\leq j\leq f_i$, job $P^{(i)}_j \in S_i$ is given $1/K$ to $M_i$ and $(1-1/K)$ to $N^{(i)}_j\in \M_i$.
			\item For $1\leq i\leq K$, job $t\in T_i$ is given $1/f_i$ to the $f_i$ machines of $\M_i$.
		\end{enumerate}
		This completes the description of the feasible solution.
	\end{proof}
	\begin{lemma}
		The integral optimum must violate some machine by factor $\Omega(K)$.
	\end{lemma}
	\begin{proof}
		Lets take machines in $\M_i$. Recall all machines here have demand of $\frac{1}{f_iK^i}$ and cardinality constraint of $1$.
		Thus in the integral optimum, they {\bf must} get one job which is either big, or of type $i$ or larger.
		Now, the total number of jobs of type $j > i$ are
		\[
		\sum_{j>i} f_j(1+1/K) = (K+1)K^{2K} \sum_{j > i} K^{-2j} \leq  (K+1)K^{2K}K^{-2i} \sum_{\ell=1}^\infty K^{-2\ell} = O\left(f_i/K\right)
		\]
		So, at least $(1 - \Theta(1/K))f_i$ of the machines in $\M_i$ get a job of type $i$ (or a big job but lets assume for now this don't happen -- can be ma).
		Therefore, the number of type $i$ jobs left after satisfying machines $(M_0,\ldots,M_K)$ are only $\Theta(f_i/K)$.

		Now take a machine $M_i$. We have $f(M_i) = f_i$ and $D(M_i) = 1/K^i$.
		First note that jobs of type $j < i$ are ``useless'' for $M_i$. Any $f_i$ of them (best to take them of type $(i-1)$)  gives capacity $f_i\cdot c_{i-1}  = \frac{f_i}{f_{i-1}K^{i-1}} = \frac{1}{K^{i+1}} =  \frac{1}{K}\cdot D(M_i)$. So any subset of these jobs that can fit in $M_i$ gives capacity $\leq D(M_i)/K$.
		On the other hand, the total capacity of jobs remaining from type $j \geq i$ is  $\sum_{j\geq i} \Theta(f_j/K)\cdot \frac{1}{f_jK^j} = \Theta(1/K)\sum_{j\geq i} \frac{1}{K^j} = \Theta(D(M_i)/K)$.
		
		Therefore, any machine $M_i$ can't get more than $D(M_i)/K$ from the ``small'' jobs. But then they all can't get big jobs.
	\end{proof}
	The above two lemmas prove~\Cref{fthm:conf-ig} after noting that $K = \Theta(\log n/\log\log n) = \Theta(\log D/\log \log D)$ where $n$ is either the number of machines of jobs and $D$ is the ratio of $D_{\max}/D_{\min}$.

\begin{theorem} \label{fthm:supply-bad-cckp}
There cannot exist $\alpha$-approximate supply polyhedra (or convex sets) for $\alpha < \frac{\log D}{\log \log D}$ for \cckp instances.
\end{theorem}

\begin{proof}
The proof follows from the instance constructed in the above~\Cref{fthm:conf-ig}. Indeed note that we can express the supply vector of the instance as $(1-p) {\bf s}_1 + p {\bf s}_2$. Here ${\bf s}_1$ denotes the following supply vector: there are $K+1$ big jobs with size $1$, and there are $f_i$ jobs of size $c_i$ for all $1 \leq i \leq K$. Similarly $\bf{s}_2$ denotes the following supply vector: there is $1$ big job with size $1$, and there are $2f_i$ jobs of size $c_i$ for all $1 \leq i \leq K$. Finally, the value $p$ is set to $1-1/K$.

Now, we will show that for both ${\bf s}_1$ and ${\bf s}_2$ are feasible supply vectors. Indeed, for ${\bf s}_1$, we will assign the large jobs to the large machines $M_0, M_1, \ldots, M_K$. Then we will use the $f_i$ jobs of size $c_i$ to satisfy the $f_i$ machines in class ${\cal M}_i$.

Likewise, for ${\bf s}_2$, we will assign the one large job to the large machine $M_0$. Then, for machine $M_i$, we will assign $f_i$ items of size $c_i$. Finally, we will assign the remaining $f_i$ jobs of size $c_i$ to satisfy the $f_i$ machines in class ${\cal M}_i$.

But now, note that for the resulting average supply vector, we proved in~\Cref{fthm:conf-ig} that any assignment must violate some demand by a factor of $\Theta(\log D/\log \log D)$, thus proving the theorem.
\end{proof}
%\end{proof}

%\newpage
\section{QPTAS for \cckp: Proof of Theorem~\ref{fthm:q}} \label{fsec:qptas}
\def\pv{\mathbf{b}}
\renewcommand{\dem}{\mathsf{cap}}
Let the instance $\calI$ given to us have $m$ machines with demands $D_1,\ldots,D_m$ and cardinality constraints $f_1,\ldots,f_m$, and $n$ jobs with capacities
$c_1,\ldots,c_n$ (with $1$ copy of each job by taking possible duplicates). Let $D:= D_\mathrm{max}/D_\mathrm{min}$, which we assume to be $\leq \poly(n)$.
Fix $\epsilon> 0$. Our goal is to either prove there is no feasible solution, or find an assignment
giving each machine $i$ a capacity of $\geq D_i(1-O(\epsilon))$.
We start with a lemma which states that finding solutions satisfying cardinality constraints approximately suffices.
\begin{lemma}\label{flem:appx-feas}
	Given an assignment of jobs such that the load on any machine $i$  is at least $D_i(1 - \epsilon_1)$ such that machine $i$ gets $\leq (1+\epsilon_2)f_i$ jobs,
	we can find another assignment which satisfies the cardinality constraints and the load of any machine $i$ is $\geq D_i(1 - \epsilon_3)$ for $\epsilon_3 < 2\epsilon_1+\epsilon_2$.
\end{lemma}
\begin{proof}
For every machine $i$, let $S_i$ be the jobs currently allocated to it. We may assume $\eps_1 f_i \geq 1$, otherwise $|S_i| \leq f_i$.
Remove the $\floor{2\eps_1 f_i} \geq \eps_1 f_i$ least capacity jobs to obtain the set $S'_i$. Note that the total capacity of $S'_i$ is at least $(1-2\eps_1)$ times capacity of $S_i$, and therefore
at least $(1-2\epsilon_1 - \epsilon_2)D_i$.
\end{proof}
\noindent
\paragraph{Input Modification and Grouping.}
We now modify the data so that everything is rounded to the nearest power of $(1+\epsilon)$. More precisely we round  $f_i$  to the {\em smallest} power of $(1+\epsilon)$ larger than the original value and $D_i$ to the largest power of $(1+\epsilon)$ smaller than the original value.
If the original instance had a feasible solution, then so does the modified instance.
For technical reasons, we round $c_p$ to the smallest value of the form $\epsilon(1+\epsilon)^t$ larger than the original value.
Let $J_p$ be the set of jobs with modified capacity $c_p = \eps (1+\epsilon)^p$, and let $n_p = |J_p|$. Furthermore, armed with Lemma~\ref{flem:appx-feas}, any $(1-\epsilon)$-approximate solution to the new instance gives an $(1-O(\epsilon))$-approximate solution to the original instance.  \smallskip

We now divide the machines into groups. For $0\leq r\leq O(\log n)$ and $0\leq s\leq O(\log n)$, let $M^{(r,s)}$ be the number of machines with $D_i = (1+\epsilon)^r$ and $f_i = (1+\epsilon)^s$.
%\[
%M^{(r,s)} := \{i: (1+\epsilon)^{r-1} \le D_i  \le (1+\epsilon)^r, ~~\textrm{and}~~ (1+\epsilon)^{s-1} \le f_i  \le (1+\epsilon)^s  \}
%\]
%Throughout the paper we fix $\delta,\eps > 0$. As is standard, we assume we have a guess $T$ of the optimum; given such a guess we would either prove $\opt > T$ or obtain a feasible schedule of makespan at most $(1+\eps)T$.
%Subsequently, binary search would provide a schedule of makespan $\leq \opt(1+\eps)$. \smallskip
%
\def\sm{\mathsf{small}}
Call a job $p$  big  for machine $i$ if $c_p \geq \eps D_i$. If $i\in M^{(r,s)}$, then $p$ lies in the set $J_r \cup J_{r+1} \cup \cdots$.
Otherwise, $p$ is small for machine $i$.
We define a bipartite graph $H$ with jobs and machines on either side, with an edge $(i,p)$ iff $p$ is small for $i$.

For every $0\leq r,s \leq O(\log n)$, we define a set of {\em feasible configurations} $\Phi^{(r,s)}$. These consist of vectors $\phi \in \Z^K_{\geq 0}$ for $K=O(1/\epsilon)$ corresponding to big jobs assigned to machines $i$ in $M^{(r,s)}$.
To be precise, $\phi_k$ is supposed to count the number of jobs with $c_p = \eps (1+\epsilon)^{r+k}$ contained in the configuration $\phi$.
The last coordinate $\phi_K$ counts the number of jobs $p$ with $c_p > (1+\epsilon)^r$.
Let
$\dem(\phi) := \sum_{k=0}^K \phi_k \epsilon(1+\epsilon)^{r+k}$ be the total load of the configuration and $|\phi| = \sum_{k=0}^K \phi_k$ be its cardinality.
We let $\Phi^{(r,s)}$ be the collection of feasible minimal configurations, that is, $\phi$'s with (a) $|\phi|\leq (1+\epsilon)^s$ and (b) either $\dem(\phi)\leq (1+\epsilon)^r$ or $\dem(\phi) > (1+\epsilon)^r$ and $\dem(\phi') \leq (1+\epsilon)^r$ for any $\phi'$ obtained by decreasing any positive coordinate of $\phi$ by exactly $1$. Note that $|\Phi^{(r,s)}| \leq N_0 = (1/\epsilon)^{(1/\epsilon)}$. Also note that in any optimal solution, each machine $i\in M^{(r,s)}$ does get one configuration from $\Phi^{(r,s)}$.
Our algorithm constructs these classes and arbitrary numbers them. The $t$th member of $\Phi^{(r,s)}$ is denoted as $\phi^{(r,s,t)}$.% where $r,s$ would be clear from context.

%if $p_j > \eps T$, and small otherwise. Let $J_\sm$ 	be the set of small jobs.
%For any $1\leq k\leq K = O(1/\eps)$, let $J_k$ denote the big jobs with capacities $p_j \in \eps T\cdot \left((1+\eps)^{k-1},(1+\eps)^k\right]$.
%A configuration $\phi \in \Z^K_{\ge 0}$ is a vector where $\phi_k$ indicates the number of jobs of chosen from $J_k$. The total load $\dem(\phi)$ of a configuration is defined to be $\sum_{k=1}^K \phi_k (1+\eps)^{k-1}\eps T$ -- note that this is a {\em lower bound} on the actual load but is within a multiplicative $(1+\eps)$-factor. Finally, we let $|\phi| = \sum_{k=1}^K \phi_k$, the cardinality of the configuration. A configuration $\phi$ is {\em allowed} if $\phi_k \leq 1/\epsilon$. The number of allowed configurations is denoted as $N_0 \leq (1/\epsilon)^{(1/\epsilon)}$.
%We call this set $\Phi$ which we order in an arbitrary but fixed manner. The $i$th member is denoted as $\phi^{(i)}$.
%
%\smallskip
%Call a machine $i$ large if $f_i > 1/\eps^2$. Let $M^{(0)}$ denote the set of large machines, and let $m_0 = |M^{(0)}|$.
%Every other machine is small and these are partitioned into $N_1 = O(1/\eps \ln(1/\eps))$ groups as follows. For $1\leq \ell \leq N_1$, define
%$M^{(\ell)} := \{i: (1+\eps)^{\ell-1} < f_i \leq (1+\epsilon)^\ell\}$. Let $m_\ell := |M^{(\ell)}|$.
%
%\subsection{MILP Formulation}
\paragraph{Enumeration.}
For every $0\leq r,s \leq O(\log n)$ and $1\le t\le N_0$, we {\em guess}  the  integer $\pv^{(r,s)}_t \in \Z_{\geq 0}$ which indicates the number of machines in $M^{(r,s)}$ who are allocated the configuration $\phi^{(t)}$.
%That is, the number of big jobs from $J_k$, for $1\le k\le K$ assigned to these machines is $\phi^{(t)}_k$. Note that the number of variables is at most $N_0 N_1 = 2^{O(1/\eps \ln(1/\epsilon))}$.
%The number of such integer variables is $d := O(\log^2 n)2^{\tilde{O}(1/\epsilon)}$. For every machine $i$ and every job $p$ small for $i$, we have a variable $z_{ip} \in [0,1]$. This indicates whether small job $p$ is allocated to machine $i$.
These guesses must satisfy
%The first constraint states that the total number of machines need to be respected.
\begin{equation}\label{feq:milp1}
\forall 0\leq r,s \leq O(\log n), ~~~~ \sum_{t=1}^{N_0} \pv^{(r,s)}_t = |M^{(r,s)}|
\end{equation}
The number of such guesses is $\leq \prod_{r,s} |M^{(r,s)}|^{N_0} \leq C_\eps^{O(\log^3 n)}$ for some constant $C_\epsilon$ which is double-exponential in $(1/\epsilon)$.
Since machines in $M^{(r,s)}$ are all equivalent (in terms of demand and cardinality constraint), by symmetry we can assign the $\pv^{(r,s)}_t$ copies of $\phi^{(r,s,t)}$ as we like.
For a guess to be feasible, for every job of type $p$, at most $n_p$ copies must be used up in the guessed configurations.
For every guess we get a residual problem on the bipartite graph $H$.
Let $n'_p$ be the remaining number of jobs of type $p$.
Let $D'_i$ be the residual demand of machine $i$, that is, $D_i - \dem(\phi)$ where $\phi$ is allocated to it by the guess.
Let $f'_i$ be the residual cardinality constraint, that is, $f'_i = f_i - |\phi|$.
\paragraph{Rounding.}
The remaining copies of jobs must satisfy the residual demand. For this we simply write the assignment LP\eqref{feq:asslp1}-\eqref{feq:asslp4} which we rewrite below.
\begin{alignat}{4}
	& \quad \forall p,   &&\quad  \textstyle \sum_{i\sim p} z_{ip} \leq n'_p \label{feq:assgnlp1} \\
	& \quad \forall i\in [m] ,  &&\quad  \textstyle \sum_{p\sim i} c_pz_{ip}  \geq D'_i \label{feq:assgnlp2}\\
	& \quad \forall i\in [m], && \quad \textstyle \sum_{p\sim i} z_{ip}  \leq  f'_i \label{feq:assgnlp3}
\end{alignat}
where $i\sim p$ implies $c_p\leq \eps D_i$. If the residual LP has no solution, then our guess of big configurations is infeasible.
We are also guaranteed some guess is correct and we get a feasible solution to above LP.
Therefore, we apply Theorem~\ref{fthm:shmoystardos} to get an integral solution $\zz_{ip}$ satisfying \eqref{feq:assgnlp1},\eqref{feq:assgnlp3},
and $\forall i\in [m], ~~ \sum_{p\sim i} c_p\zz_{ip}  \geq D'_i - \epsilon D_i$.
Therefore in all every machine receives capacity $\geq D_i(1-\epsilon)$. The total running time is dominated by the enumeration step.
This proves~\Cref{fthm:q}.

\section{Integrality Gap for Non-Uniform Santa Claus Problem}\label{sec:app-bsig}
We reproduce the integrality gap example for the configuration LP by Bansal and Sviridenko~\cite{BansalS06} for the general max-min allocation problem, and point out how their instance is in fact a $Q|restr|C_{min}$ instance.
Fix integer $K$. There are $K$ machines with demand $D_i = K$; these are the large machines $L = \{M_1,M_2,\ldots,M_K\}$. There are $K-1$ large jobs with $c_j = K$ which can only be assigned to the machines in $L$.
Let $J_B$ be the set of large jobs. There are $K^2$ small machines each with $D_i = 1$; these machines are distributed in $K$ classes where the $i$th class $\cC_i$ contains $K$ small machines. We let $m^{(i)}_k$ denote the $k$th machine in $\cC_i$, for $1\leq k\leq K$.
There are $K^2 + K$ small jobs with $c_j = 1$. These jobs are partitoned into $K$ classes with $i$th class $\cJ_i$ containing $K+1$ small jobs. Each class $\cJ_i$ has one ``public'' job $j^{(i)}_0$ which can be assigned to any machine $m^{(i)}_k \in \cC_i$
 and $K$ ``private'' jobs $j^{(i)}_k$, $1\leq k\leq K$ where $j^{(i)}_k$ can be assigned to only $m^{(i)}_k \in \cC_i$. Furthermore all the private jobs $j^{(i)}_k\in \cJ_i$ can be assigned to the large machine $M_i \in L$. This completes the description of the instance.
Note that the number of machines and jobs are $\Theta(K^2)$.

The integral optimum solution must give one machine $i$ capacity $\leq D_i/K$. Indeed, at least one large machine $M_i$ will not receive a job in $J_B$.
The only other jobs available to $M_i$ are the private jobs in $\calJ_i$. Suppose we allocate two such jobs to $M_i$; wlog these are $j^{(i)}_1$ and $j^{(i)}_2$. Now note
that the machines $m^{(i)}_1$ and $m^{(i)}_2$ have only job $j^{(i)}_0$ which can be assigned to them; and so one of them would get capacity $0$.
Therefore, the machine $M_i$ can receive only one job $j^{(i)}_k$ giving it total capacity $\leq D_i/K$.

On the other hand the configuration LP is feasible. Every large machine $M_i$ gets $z(M_i, j) = 1/K$ for all large jobs $j\in J_B$ and $z(M_i,\{j^{(i)}_1,\ldots, j^{(i)}_K\}) = 1/K$ for the set of private jobs in $\cJ_i$.
For all $1\leq i,k\leq K$, every machine $m^{(i)}_k$ receives $z(m^{(i)}_k,j^{(i)}_k) = 1-1/K$ and $z(m^{(i)}_k,j^{(i)}_0) = 1/K$. One can check all the jobs are fractionally assigned exactly.

\section{Conclusion}
In this paper we introduced and studied the \mckc problem, and highlighted its connection to an interesting special case of the max-min allocation problems, namely \cckp. In our main result, we showed, using a decomposition theorem and the notion of supply polyhedra, a logarithmic approximation for \cckp, using which we showed a bicriteria $(O(1),O(\log n))$-approximation for \mckc. We believe designing polynomial-time $O(1)$-approximations for \cckp and bicriteria $(O(1),O(1))$ algorithms for \mckc are very interesting open problems.

\bibliographystyle{abbrv}
\bibliography{mckc}

\end{document}